\documentclass[sigconf]{aamas} 

\usepackage{balance} 
\usepackage{subcaption}
\usepackage[ruled,noend,linesnumbered]{algorithm2e}
\usepackage{tikz}
\usetikzlibrary{arrows,positioning,automata}

\renewcommand{\citeyear}[1]{\citealp{#1}}

\setcopyright{ifaamas}
\acmConference[AAMAS '24]{Proc.\@ of the 23rd International Conference
on Autonomous Agents and Multiagent Systems (AAMAS 2024)}{May 6 -- 10, 2024}
{Auckland, New Zealand}{N.~Alechina, V.~Dignum, M.~Dastani, J.S.~Sichman (eds.)}
\copyrightyear{2024}
\acmYear{2024}
\acmDOI{}
\acmPrice{}
\acmISBN{}

\acmSubmissionID{54}

\title{Coalition Formation with Bounded Coalition Size}

\author{Chaya Levinger}
\affiliation{
  \institution{Ariel University}
  \city{Ariel}
  \country{Israel}}
\email{chayal@ariel.ac.il}

\author{Noam Hazon}
\affiliation{
  \institution{Ariel University}
  \city{Ariel}
  \country{Israel}}
\email{noamh@ariel.ac.il}

\author{Sofia Simola}
\affiliation{
  \institution{Technische Universität Wien}
  \city{Vienna}
  \country{Austria}}
\email{sofia.simola@tuwien.ac.at}

\author{Amos Azaria}
\affiliation{
  \institution{Ariel University}
  \city{Ariel}
  \country{Israel}}
\email{amos.azaria@ariel.ac.il}

\begin{abstract}
In many situations when people are assigned to coalitions, the utility of each person depends on the friends in her coalition. Additionally, in many situations, the size of each coalition should be bounded. This paper studies such coalition formation scenarios in both weighted and unweighted settings.
Since finding a partition that maximizes the utilitarian social welfare is computationally hard, we provide a polynomial-time approximation algorithm. 
We also investigate the existence and the complexity of finding stable partitions. Namely, we show that 
the Contractual Strict Core (CSC) is never empty, but the Strict Core (SC) of some games is empty. 
Finding partitions that are in the CSC is computationally easy, but even deciding whether an SC of a given game exists is NP-hard. 
The analysis of the core is more involved. 
In the unweighted setting, we show that when the coalition size is bounded by $3$ the core is never empty, and we present a polynomial time algorithm for finding a member of the core. However, for the weighted setting, the core may be empty, and we prove that deciding whether there exists a core is NP-hard.

\end{abstract}


\newcommand{\BibTeX}{\rm B\kern-.05em{\sc i\kern-.025em b}\kern-.08em\TeX}

\makeatletter
\gdef\@copyrightpermission{
	\begin{minipage}{0.3\columnwidth}
		\href{https://creativecommons.org/licenses/by/4.0/}{\includegraphics[width=0.90\textwidth]{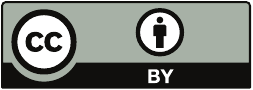}}
	\end{minipage}\hfill
	\begin{minipage}{0.7\columnwidth}
		\href{https://creativecommons.org/licenses/by/4.0/}{This work is licensed under a Creative Commons Attribution International 4.0 License.}
	\end{minipage}
	\vspace{5pt}
}
\makeatother

\newenvironment{procedure1}[1][htb]
  {
  \begin{algorithm}[#1]%
  }{\end{algorithm}}

\newtheorem{definition}{Definition}[section]
\newtheorem{theorem}{Theorem}
\newtheorem{lemma}[theorem]{Lemma}
\newtheorem{claim}[theorem]{Claim}

\newcommand{\wa}{7M}
\newcommand{\wk}{(k - 2)15M}
\newcommand{\wt}{6M}
\newcommand{\TDSRM}{\textsc{Metric-3DSR}}
\newcommand{\coalS}{S}
\newcommand{\partitionS}{P}
\newcommand{\Wset}{D}
\newcommand{\Wvertex}{d}
\newcommand{\uS}[2]{u(#1,#2)}
\newcommand{\uCPS}[3]{u(#1, {#3}^{-#2})}
\newcommand{\uCS}[2]{W(#1, #2)}

\begin{document}


\pagestyle{fancy}
\fancyhead{}


\maketitle 


\section{Introduction}

Suppose that a group of travelers, located at some origin, would like to reach the same destination, and later return. Each of the travelers has her own vehicle; but each traveler has a preference related to who will be with her in the vehicle. Namely, each traveler would rather share a vehicle with as many as possible of her friends during the ride, and thus the utility of each traveler is the number of friends traveling with her. However, the vehicles have a limited capacity; this capacity can either be a physical constraint of the vehicles, or the maximal number of travelers willing to travel together. How should the travelers be assigned to vehicles in order to maximize the social welfare (the sum of all travelers' utilities)? Can the travelers be organized in such a way that no subgroup of travelers will want to leave their current group and join together?
Similar questions arise when assigning students to dormitories, colleagues to office-rooms and workers to project teams. In these settings, it might be that the utility of each person does not depend only on the number of friends, but also on the intensity of friendship. 

This set of problems falls within hedonic games \cite{dreze1980hedonic}, in which a set of agents are partitioned into coalitions, and the utility for each agent depends only on the coalition that she is a member of.
Additively Separable Hedonic Games (ASHGs) \cite{bogomolnaia2002stability} are a special type of hedonic games, in which each agent has a value for any other agent, and the utility she assigns to a coalition is the sum of the values she assigns to its members. In ASHGs there is usually no restriction on the number of agents that are allowed to belong to a coalition. However, in our group of travelers example, the vehicles have physical capacity, and thus, there is an upper bound on the size of each coalition. Despite this restriction being natural, it is scarcely studied in the domain of hedonic games. 

In this paper, we study hedonic games with bounded coalition size. Specifically, we concentrate on symmetric ASHGs, in which the value an agent assigns to another agent is non-negative and it is equal to the value that the other agent assigns to her; we refer to these settings as the weighted settings.
We also study simple symmetric ASHGs, in which the value an agent assigns to other agents is either $0$ or $1$; we refer to these settings as the unweighted settings. 
These models capture many situations, such as social and friendship relations. We begin by studying the problem of finding a partition that maximizes the utilitarian social welfare. Since this problem is computationally hard for any coalition size bound $k>2$, even in the unweighted setting, we provide a polynomial-time approximation algorithm. We prove that in the unweighted setting, the algorithm has an approximation ratio of $\frac{1}{k-1}$. In the weighted setting, the algorithm has an approximation ratio of $\frac{1}{k}$ when $k$ is odd, and $\frac{1}{k-1}$ when $k$ is even.

We then study stability aspects of the problem. That is, we investigate the existence 
and the complexity of finding stable partitions. Namely, we show that 
the Contractual Strict Core (CSC) is never empty, but the Strict Core (SC) of some games is empty. 
Finding partitions that are 
in the CSC is computationally easy, but even deciding whether an SC of a given game exists is hard. 
The analysis of the core is more involved. 
In the unweighted setting, we show that for k=$3$ the core is never empty, and we present a polynomial time algorithm for finding a member of the core.
For $k>3$, it is unclear whether the core can be empty, and how to find a partition in the core.
Indeed, we show in simulation over 100 million games that a simple heuristic always finds a partition that is in the core.
For the weighted setting, the core may be empty even when $k=3$, and we prove that for any $k \geq 3$, deciding whether there exists a partition in the core is NP-hard.

To summarize, the contribution of this work is a systematic study of additively separable hedonic games with bounded coalition size. Namely, we provide an approximation algorithm for maximizing the utilitarian social welfare and study the computational aspects of several stability concepts.

\section{Related Work}

Dreze and Greenberg [\citeyear{dreze1980hedonic}] initiated the study of hedonic games,
in which the utility for each agent depends only on the coalition that she is a member of. Stability concepts of hedonic games were further analyzed in \cite{banerjee2001core} and \cite{cechlarova2001stability}. For more details, see the survey of Aziz et al. [\citeyear{aziz2016hedonic}].
A special case is Additively Separable Hedonic Games (ASHGs) \cite{bogomolnaia2002stability}, in which each agent has a value for any other agent, and the utility she assigns to a coalition is the sum of the values she assigns to its members. The computational aspects of ASHGs are analyzed in \cite{deng1994complexity,ballester2004np,olsen2009nash,sung2010computational,aziz2013computing,Bachrach2013,bilo2019optimality}. 
None of these works imposed any restriction on the size of the coalitions.

Indeed, there are few papers that impose a restriction on the size of the coalitions. 
Wright and Vorobeychik [\citeyear{wright2015mechanism}] study a model of ASHG where there is an upper bound on the size of
each coalition. Within their model, they propose a strategyproof
mechanism that achieves good and fair experimental
performance, despite not having a theoretical guarantee.
Flammini et al. [\citeyear{flammini2021online}] study the online partition problem. Similar to our work, they also consider the scenario that the coalitions are bounded by some number. They consider two cases for the value of a coalition, the sum of the weights of its edges, which is similar to our work, and the sum of the weights of its edges divided by its size. However, in both cases they only consider the online version, i.e., the agents arrive sequentially and must be assigned to a coalition as they arrive. This assignment cannot be adapted later on, and must remain. They show that a simple greedy algorithm achieves an approximation ratio of $\frac{1}{k}$ when the value of the coalition is the sum of the weights. 
Cseh et al. [\citeyear{cseh2019pareto}] require the partition to be composed of exactly $k$ coalitions, and also assume a predefined set of size constraints. Each coalition is required to exactly match its predefined size. They study the complexity of finding a Pareto optimal partition, as well as the complexity of deciding whether a given partition is Pareto optimal.
Bilò et al. [\citeyear{bilo2022hedonic}] consider the same settings as Cseh et al.
Since classical stability notions are infeasible in their setting, they study three different types of swap stability, and analyze the existence, complexity, and efficiency of stable outcomes. 
Note that almost all other works analyzing ASHGs assume that an agent may assign a negative value to another agent. Otherwise, since they do not impose any restrictions on the coalition size, the game becomes trivial, as the grand coalition is always an optimal solution. We found two exceptions that restrict the value each agent assigns to other agents to be either $0$ or $1$. Namely, 
Sless et al. [\citeyear{sless2018forming}] study the setting in which 
 the agents must be partitioned into exactly $k$ coalitions, without any restriction on each coalition's size. Li et al. [\citeyear{li2023partitioning}] study the setting in which the agents must be partitioned into exactly $k$ coalitions that are almost equal in their size.

\section{Preliminaries}
Let $V=\{v_1, ..., v_n\}$ be a set of agents, and let $G(V,E)$ be a weighted undirected graph representing the social relations between the agents. For every edge $e \in E$, the weight of the edge, $w(e)$, is positive. In the unweighted setting, all weights are set to $1$. 
A $k$-bounded coalition is a coalition of size at most $k$. A $k$-bounded partition $P$ is a partition of the agents into disjoint $k$-bounded coalitions. Given a coalition $S \in P$, and $v\in S$, let $N(v,S)$ be the set of immediate neighbors of $v \in V$ in $S$, i.e., $N(v,S) = \{u \in S : (v,u) \in E\}$. Let $W(v,S)$ be the sum of weights of immediate neighbors of $v \in V$ in $S$, i.e., $W(v,S) = \sum_{u \in N(v,S)}{w((v,u))}$. Note that in the unweighted setting, $W(v,S) = |N(v,S)|$.
An \textit{additively separable hedonic game with bounded coalition size} is a tuple $(G,k)$, where for every $k$-bounded partition $P$, coalition $S \in P$, and $v \in S$, the agent $v$ gets utility $W(v,S)$. We denote the utility of $v$ given a $k$-bounded partition $P$, by $u(v,P)$. Given a tuple $(G,k)$, the goal is to find a $k$-bounded partition $P$ that satisfies efficiency or stability properties.

We consider the following efficiency or stability concepts: 
\begin{itemize}
    \item The utilitarian social welfare of a partition $P$, denoted $u(P)$, is the sum of the utilities of the agents. That is, $u(P) = \sum\limits_{v\in V} u(v,P)$.
    A \emph{MaxUtil} $k$-bounded partition $P$ is a partition with maximum $u(P)$.
    \item A $k$-bounded coalition $S$ is said to \emph{strongly block} a $k$-bounded partition $P$ if for every $v \in S$, $W(v,S)>u(v,P)$.
    A $k$-bounded partition $P$ is in the \emph{Core} if it does not have any strongly blocking $k$-bounded coalitions.
    \item A $k$-bounded coalition $S$ is said to \emph{weakly block} a $k$-bounded partition $P$ if for every $v \in S$, $W(v,S)\geq u(v,P)$, and there exists some $v \in S$ such that $W(v,S) > u(v,P)$.
    A $k$-bounded partition $P$ is in the \emph{Strict Core (SC)} if it does not have any weakly blocking $k$-bounded coalitions.
    \item Given a partition $P$ and a set $S$, let $P^{-S}$ be the partition when $S$ breaks off. That is, $P^{-S}=\{S\} \cup \bigcup\limits_{C\in P} \{C\setminus S\}$.
    A $k$-bounded partition $P$ is in the \emph{Contractual Strict Core (CSC)} if for any weakly blocking $k$-bounded coalition $S$, there exists at least one agent $v$ such that $u(v,P^{-S})<u(v,P)$.
\end{itemize}


\section{Efficiency} 
We begin with the elementary concept of efficiency, which is to maximize the utilitarian social welfare.

\begin{definition}[MaxUtil problem]
Given a coalition size limit $k$ and a graph $G$, find a MaxUtil $k$-bounded partition.


\end{definition}
For example, given the unweighted graph in Figure \ref{fig:Graph} and a coalition size limit $k=3$, let $P = \{\{v_1,v_3,v_6\}, \allowbreak \{v_2,v_4,v_7\}, \{v_5,v_8\}\}$, shown in Figure~\ref{fig:Partition}. The utilitarian social welfare of this partition, $u(P)$ equals $14$. Indeed, this is an optimal $3$-bounded partition, since there is no other $3$-bounded partition with higher social welfare.
Clearly, the decision variant of the MaxUtil problem is to decide whether there exists a $k$-bounded partition with a utilitarian social welfare of at least $\upsilon$.


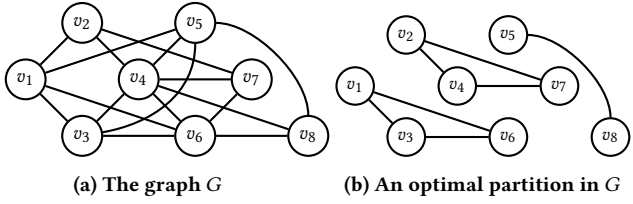
\begin{figure}
     \centering
     \begin{subfigure}{0.45\columnwidth}
         \centering
         \begin{tikzpicture}[node distance={12.5mm}, thick, main/.style = {draw, circle, scale=0.85}] {Graph}
            \node[main] (a) {$v_1$}; 
            \node[main] (b) [above right of=a] {$v_2$}; 
            \node[main] (c) [below right of=a] {$v_3$}; 
            \node[main] (d) [above right of=c] {$v_4$}; 
            \node[main] (e) [above right of=d] {$v_5$}; 
            \node[main] (f) [below right of=d] {$v_6$}; 
            \node[main] (g) [below right of=e] {$v_7$}; 
            \node[main] (h) [below right of=g] {$v_8$}; 
            \draw (a) -- (b); 
            \draw (a) -- (c); 
            \draw (a) -- (e); 
            \draw (b) -- (d); 
            \draw (c) -- (d); 
            \draw (e) -- (d); 
            \draw (e) to [out=270, in=10, looseness=1] (c); 
            \draw (f) -- (d); 
            \draw (g) -- (f); 
            \draw (g) -- (d); 
            \draw (f) -- (a); 
            \draw (g) -- (b);
            \draw (e) to [out=360, in=90, looseness=0.75] (h); 
            \draw (h) -- (d);
            \draw (h) -- (f);
            \draw (c) -- (f);
        \end{tikzpicture}
        \caption{The graph $G$}
        \label{fig:Graph}
     \end{subfigure}
     \hfill
     \begin{subfigure}{0.5\columnwidth}
         \centering
         \begin{tikzpicture}[node distance={12mm}, thick, main/.style = {draw, circle, scale=0.8}] {Graph}
            \node[main] (a) {$v_1$}; 
            \node[main] (b) [above right of=a] {$v_2$}; 
            \node[main] (c) [below right of=a] {$v_3$}; 
            \node[main] (d) [above right of=c] {$v_4$}; 
            \node[main] (e) [above right of=d] {$v_5$}; 
            \node[main] (f) [below right of=d] {$v_6$}; 
            \node[main] (g) [below right of=e] {$v_7$}; 
            \node[main] (h) [below right of=g] {$v_8$}; 
            \draw (a) -- (c); 
            \draw (f) -- (a); 
            \draw (c) -- (f);
            \draw (b) -- (d);
            \draw (g) -- (d); 
            \draw (g) -- (b);
            \draw (e) to [out=360, in=90, looseness=0.75] (h); 
        \end{tikzpicture}
        \caption{An optimal partition in $G$}
        \label{fig:Partition}
    \end{subfigure}
    \caption{An example for the MaxUtil problem where $k=3$.}
\end{figure}




\subsection{Approximation of the MaxUtil Problem}
The MaxUtil problem when $k=2$ is equivalent to the maximum (weight) matching problem, and thus it can be computed in polynomial time \cite{edmons1965paths}.
However, our problem becomes intractable when $k \geq 3$ even in the unweighted setting. 
Indeed, the decision variant of the MaxUtil problem in the unweighted setting is equivalent to the graph partitioning problem as defined by Hyafil and Rivest [\citeyear{HR73}], which they show to be $NP$-Complete.
Therefore, we provide the Match and Merge (MnM) algorithm (Algorithm \ref{alg:MnM}), which is a polynomial-time approximation algorithm for any $k \geq 3$. 
The algorithm consists of $k-1$ rounds. Each round is composed of a matching phase followed by a merging phase.
Specifically, in round $l$ MnM computes a maximum (weight) matching, $M_l \subseteq E_l$, for $G_l$ (where $G_1 = G$). In the merging phase, MnM creates a graph $G_{l+1}$ that includes a unified node for each pair of matched nodes. The graph $G_{l+1}$ also includes all unmatched nodes, along with their edges to the unified nodes (lines \ref{line:marge_edges_s}-\ref{line:marge_edges_e}).
Clearly, each node in $V_l$ is composed of up-to $l$ nodes from $V_1$.
Finally, MnM returns the $k$-bounded partition, $P$, of all the matched sets.
For example, given the graph $G_1$ in Figure \ref{fig:G_1} and $k=4$, the algorithm finds a maximum matching $M_1 = \{(v_1,v_2),(v_3,v_4)\}$ shown in Figure \ref{fig:M_1}.  It then creates the graph $G_2$, as shown in Figure \ref{fig:G_2}, and finds a maximum matching for it, $M_2 = \{(v_{3,4},v_5)\}$ shown in Figure \ref{fig:M_2}. It then creates the graph $G_3$, as shown in Figure \ref{fig:G_3}, and finds a maximum matching for it, $M_3=\{(v_{3,4,5}, v_6)\}$. 
Finally, MnM created the graph $G_4$, as shown in Figure \ref{fig:G_4}, and returns the $4$-bounded partition $P={\{v_1,v_2\},\{v_3,v_4,v_5,v_6\}}$.
We note that by the algorithm construction, a unified node $v_{i_1,...,i_l}$, is created by merging nodes $v_{i_1}$ and $v_{i_2}$, and then by merging $v_{i_1,i_2}$ and $v_{i_3}$, and so on.
 
\begin{algorithm}[tbp]
    \caption{Match and Merge (MnM)}
    \label{alg:MnM}
    \SetAlgoLined
    \textbf{Input}:
    A graph $G(V,E)$, and a limit $k$.\\
    \KwResult{A $k$-bounded partition $P$ of $V$.} 
    $G_1(V_1, E_1) \leftarrow G(V,E)$\\
    \For{$l\leftarrow1$ to $k-1$}{
        $M_l$ $\leftarrow$ maximum (weight) matching in $G_l$\\
        $G_{l+1}=(V_{l+1},E_{l+1}) \leftarrow$ an empty graph\\
        $V_{l+1} \leftarrow V_l$\\
        \For{\upshape every $(v_{i_1,...,i_l}, v_j) \in M_l$}{
            add vertex $v_{i_1,...,i_l,j}$ to $V_{l+1}$\\
            remove $v_{i_1,...,i_l}, v_j$ from $V_{l+1}$
        }
        \For{\upshape every $v_{i_1,...,i_{l+1}} \in V_{l+1}$}{\label{line:marge_edges_s}
            \For{\upshape every $v_q \in V_{l+1}$}{
                \If {\upshape $(v_{i_1,...,i_l}, v_q) \in E_l$ or $(v_{i_{l+1}}, v_q) \in E_l$}{
                    add $(v_{i_1,...,i_{l+1}}, v_q)$ to $E_{l+1}$
                }
            }
        }\label{line:marge_edges_e}
    }
    
     $P \leftarrow$ an empty partition\\
     \For{\upshape every $v_{i_1,...,i_j} \in G_k$}{
         add the set $\{v_{i_1},...,v_{i_j}\}$ to $P$
     }
    \textbf{return} P
\end{algorithm}

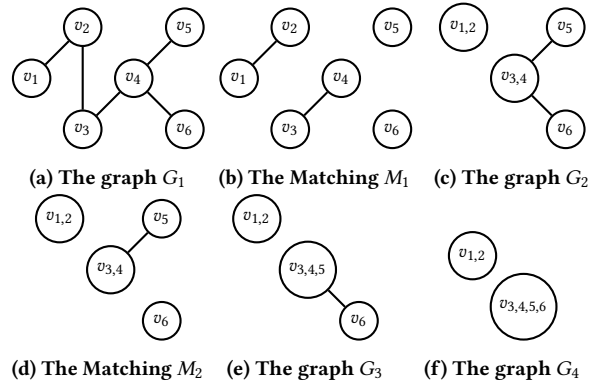
\begin{figure}[htbp]
    \centering
    \begin{subfigure}{.33\columnwidth}
        \centering
        \begin{tikzpicture}[node distance={12mm}, thick, main/.style = {draw, circle, scale = 0.8}] {Graph}
        \node[main] (1) {$v_1$}; 
        \node[main] (2) [above right of=1] {$v_2$}; 
        \node[main] (3) [below right of=1] {$v_3$}; 
        \node[main] (4) [above right of=3] {$v_4$}; 
        \node[main] (5) [above right of=4] {$v_5$}; 
        \node[main] (6) [below right of=4] {$v_6$}; 
        \draw (1) -- (2); 
        \draw (3) -- (2); 
        \draw (3) -- (4); 
        \draw (4) -- (5); 
        \draw (4) -- (6); 
        \end{tikzpicture}
        \caption{The graph $G_1$}
        \label{fig:G_1}
    \end{subfigure}
    \begin{subfigure}{.3\columnwidth}
        \centering
        \begin{tikzpicture}[node distance={12mm}, thick, main/.style = {draw, circle, scale = 0.8}] {Graph}
        \node[main] (1) {$v_1$}; 
        \node[main] (2) [above right of=1] {$v_2$}; 
        \node[main] (3) [below right of=1] {$v_3$}; 
        \node[main] (4) [above right of=3] {$v_4$}; 
        \node[main] (5) [above right of=4] {$v_5$}; 
        \node[main] (6) [below right of=4] {$v_6$}; 
        \draw (1) -- (2); 
        \draw (3) -- (4); 
        \end{tikzpicture}
        \caption{The Matching $M_1$}
        \label{fig:M_1}
    \end{subfigure}
    \begin{subfigure}{.3\columnwidth}
        \centering
        \begin{tikzpicture}[node distance={12mm}, thick, main/.style = {draw, circle, scale = 0.8}] {Graph}
        \node[main] (12) {$v_{1,2}$}; 
        \node[main] (34) [below right of=12] {$v_{3,4}$};  
        \node[main] (5) [above right of=34] {$v_5$}; 
        \node[main] (6) [below right of=34] {$v_6$}; 
        \draw (34) -- (5); 
        \draw (34) -- (6); 
        \end{tikzpicture}
        \caption{The graph $G_2$}
        \label{fig:G_2}
    \end{subfigure}
    \begin{subfigure}{.3\columnwidth}
        \centering
        \begin{tikzpicture}[node distance={12mm}, thick, main/.style = {draw, circle, scale = 0.8}] {Graph}
        \node[main] (12) {$v_{1,2}$}; 
        \node[main] (34) [below right of=12] {$v_{3,4}$};  
        \node[main] (5) [above right of=34] {$v_5$}; 
        \node[main] (6) [below right of=34] {$v_6$}; 
        \draw (34) -- (5); 
        \end{tikzpicture}
        \caption{The Matching $M_2$}
        \label{fig:M_2}
    \end{subfigure}
    \begin{subfigure}{.3\columnwidth}
        \centering
        \begin{tikzpicture}[node distance={12mm}, thick, main/.style = {draw, circle, scale = 0.8}] {Graph}
        \node[main] (12) {$v_{1,2}$}; 
        \node[main] (345) [below right of=12] {$v_{3,4,5}$};  
        \node[main] (6) [below right of=34] {$v_6$}; 
        \draw (345) -- (6); 
        \end{tikzpicture}
        \caption{The graph $G_3$}
        \label{fig:G_3}
    \end{subfigure}
    \begin{subfigure}{.3\columnwidth}
        \centering
        \begin{tikzpicture}[node distance={12mm}, thick, main/.style = {draw, circle, scale = 0.8}] {Graph}
        \node[main] (12) {$v_{1,2}$}; 
        \node[main] (3456) [below right of=12] {$v_{3,4,5,6}$};  
        \end{tikzpicture}
        \caption{The graph $G_4$}
        \label{fig:G_4}
    \end{subfigure}
    \caption{An example for Algorithm \ref{alg:MnM} where $k=4$.}
\end{figure}

\subsection{Approximation Ratio for Unweighted Setting}
We first show that MnM provides an approximation ratio of $\frac{1}{k-1}$ for the MaxUtil problem in the unweighted setting. For that end, we first prove the following lemma related to the possible edges in every $G_l$, $l>1$. Note that the indexes follow the order in which the nodes join the matched node.

\begin{lemma}
\label{lemma:5.5}
Given $\hat{v} = v_{i_1,...,i_l} \in V_l$, if there exist $v_i, v_j \in V_l$, $v_i\neq v_j$ such that $(v_i, v_{i_n}), (v_j, v_{i_m}) \in E$ for some $1 \leq n \leq m \leq l$, then $n = m$.
\end{lemma}

\begin{proof}
Observe that for every $v_i, v_j \in V_l$ where $l>1$, $(v_i,v_j) \notin E$, since $M_1$ is a maximum matching in $G_1$.  
Assume by contradiction and without loss of generality that $n<m$. If $n=1$ and $m=2$, then the path $v_i \rightarrow v_{i_n} \rightarrow v_{i_m} \rightarrow v_j$ is an $M_1$-augmenting path in $G_1$ (\cite{edmons1965paths}), contrary to the fact that $M_1$ is a maximum matching in $G_1$. Therefore, $m\geq 3$. 
 
Now, since $v_{i_m}$ joined the merged node only after the first merge stage, it must be a singleton node in $G_2$ (as well as $v_j$). In addition, since ($v_j$, $v_{i_m})\in E$, they should have been matched at the very first stage.


\end{proof}

We now present a hypothetical procedure (Procedure \ref{alg:findMatch}) that is provided with a solution to the MaxUtil problem, which is a $k$-bounded partition (of $G$) $Opt$, a graph $G_l$ (as defined in Algorithm \ref{alg:MnM}), and a corresponding round index $l$. Without loss of generality, we assume that every set $S \in Opt$ is a connected component. %
%
%
Let $O=\{v_o \mid \{v_o\} \in Opt$ and $v_o \in V_2\}$. That is, $|O|$ is the number of singletons in the partition $Opt$ that are also not matched in $M_1$. %
%
We show that Procedure \ref{alg:findMatch} finds a matching, and we provide a lower bound on the size of this matching (the number of edges in it). 
We further show that MnM is guaranteed to perform at least as well as this procedure, which, as we show, results in an approximation ratio of $\frac{1}{k-1}$ for every $k\geq3$.


\begin{procedure1}[hbtp]
    \caption{Find matching}
    \label{alg:findMatch}
    \SetAlgoLined
    \textbf{Input}:
    A $k$-bounded partition (of $G$) $Opt$, 
    and a graph $G_l=(V_l,E_l)$.\\
    \KwResult{A matching in $G_l$.}

    $R_l \leftarrow$ an empty matching\\
    \For {\upshape each $v_i \in V_l$ such that $\{v_i\} \in Opt$}
        {remove $v_i$ from $V_l$}
    \For{\upshape each $v_q \in V_l$}{\label{line:loop_k>3}
        let $\hat{v}$ be a vertex $v_{i_1,...,i_l}$ such that $(v_q, \hat{v})\in E_l$ and for a $1\leq j\leq l$, $v_q$ and $v_{i_j}$ belong to the same set in $Opt$ \\
            \For {\upshape each $v_n \neq v_q$}{ 
            \If {\upshape $(v_n, \hat{v}) \in E_l$ and exists $1\leq m\leq l$, s.t. $v_{i_m}$ and $v_n$ belong to the same set in $Opt$}{
                remove $v_n$ from $V_l$ }\label{line:remove_vn}
            }
        add $(v_q, \hat{v})$ to $R_l$\label{line:add_edge_k>3}
    }
    \textbf{return} $R_l$
\end{procedure1}

\begin{lemma}
\label{lem:Rl_size}
Procedure \ref{alg:findMatch} finds a matching, $R_l$, in the graph $G_l$, such that $|R_l| \geq (|V|-2|M_1|-\sum\limits_{i=2}^{l-1}|M_i|-|O|)/(k-1)$, where $l>1$.
\end{lemma}

\begin{proof}
We first show that Procedure \ref{alg:findMatch} finds a matching, $R_l$, in the graph $G_l$.
At each iteration of the loop in line \ref{line:loop_k>3}, we add an edge between a single node, $v_q$, and a unified node, $v_{i_1,...,i_l}$.
We consider each single node only once. Therefore, it is not possible to add a single node twice to $R_l$. Similarly, each time a unified node is added to $R_l$, every single node $v_n \neq v_q$ such that $v_{i_m}$ and $v_n$ belong to the same set in $Opt$, for some $1\leq m\leq l$, is removed from $V_l$. Therefore, a unified node is not added more than once. That is, $R_l$ is a matching in $G_l$.

We now show a lower bound on the size of $|R_l|$.
Let $V_l'=\{v_i | v_i \in V_l\}$, i.e., the set of all the single nodes in $G_l$.
In line \ref{line:remove_vn} we remove nodes only when $m=j$ (according to Lemma \ref{lemma:5.5}).
Given $\hat{v} = v_{i_1,...,i_l}$, there are at most $k-1$ different nodes, $v_{j_1},...,v_{j_{k-1}}$ that are in the same set with $\hat{v}$ in $Opt$. Therefore, in each iteration of the loop in line \ref{line:loop_k>3}, we remove at most $k-2$ single nodes in line \ref{line:remove_vn} while adding one edge to $R_l$ in line \ref{line:add_edge_k>3}.
Thus, at least $\frac{1}{k-1}$ of the single nodes in $V_l$ (who are not in $O$) are matched to a unified node. Therefore, $|R_l| \geq \frac{|V_l'|-|O|}{k-1}$.
Now, $|V_2'|=|V_1|-2|M_1|$. In addition, at each iteration $l>i>1$, $|M_i|$ single nodes are each added to a unified node. Therefore, $|V_l'|=|V_1|-2|M_1|-\sum\limits_{i=2}^{l-1}|M_i|$. 
In addition, $V=V_1$.
Overall, $|R_l| \geq (|V|-2|M_1|-\sum\limits_{i=2}^{l-1}|M_i|-|O|)/(k-1)$.
\end{proof}

\begin{theorem}
Algorithm \ref{alg:MnM} provides a solution for the MaxUtil problem with an approximation ratio of $\frac{1}{k-1}$ for every $k\geq 3$, in the unweighted setting.
\end{theorem}

\begin{proof}
Let $P$ be the $k$-bounded partition returned by Algorithm \ref{alg:MnM}.
Clearly, $u(P) \geq 2\cdot\sum\limits_{i=1}^{k-1} |M_i|$.
For every $l\geq 1$, $M_l$ is a maximum matching and thus $|M_l| \geq |R_l|$. In addition, according to Lemma \ref{lem:Rl_size}, $|R_l| \geq \frac{|V|-2|M_1|-\sum\limits_{i=2}^{l-1} |M_i| - |O|}{k-1}$. 
%
Therefore, \[u(P) \geq 2\cdot\sum\limits_{i=1}^{k-1} |M_i| = 2|M_1| + 2\cdot \sum\limits_{i=2}^{k-1} |M_i|.\]
\[\sum\limits_{i=2}^{k-1} |M_i| = |M_2| + |M_3| + ... + |M_{k-1}| \geq\] \[|M_2| + |M_3| + ... + |M_{k-2}| + |R_{k-1}| \geq |M_2| + |M_3| + ... + |M_{k-2}| + \] \[\frac{|V| - |O| - 2|M_1| - |M_2| - ... - |M_{k-2}|}{k-1} =\] \[\frac{|V| - |O| - 2|M_1|}{k-1} + \frac{k-2}{k-1} \sum\limits_{i=2}^{k-2} |M_i| \geq (1 + \frac{k-2}{k-1}) \cdot \frac{|V| - |O| - 2|M_1|}{k-1} + \] \[(\frac{k-2}{k-1})^2 \sum\limits_{i=2}^{k-3} |M_i| \geq ... \geq \]\[(1 + \frac{k-2}{k-1} + (\frac{k-2}{k-1})^2 + ... + (\frac{k-2}{k-1})^{k-3}) \cdot \frac{|V| - |O| - 2|M_1|}{k-1} + \] \[(\frac{k-2}{k-1})^{k-2} \sum\limits_{i=2}^{k-1-(k-2)} |M_i| = \sum\limits_{i=0}^{k-3}\Big((\frac{k-2}{k-1})^i \cdot \frac{|V| - |O| - 2|M_1|}{k-1}\Big).\]

That is, \[u(P) \geq 2|M_1| + 2\cdot\sum\limits_{i=0}^{k-3}\Big((\frac{k-2}{k-1})^i \cdot \frac{|V| - |O| - 2|M_1|}{k-1}\Big) =\]
%
%
%
%
\[2|M_1| + 2\cdot\frac{|V|-|O|-2|M_1|}{k-1} \cdot \frac{(\frac{k-2}{k-1})^{(k-2)}-1}{\frac{k-2}{k-1} -1} =\] \[2|M_1| + 2(|V|-|O|-2|M_1|) \cdot \frac{(\frac{k-2}{k-1})^{(k-2)}-1}{(k-1)(\frac{k-2}{k-1} -1)} = \]\[ 2|M_1| - 2(|V|-|O|-2|M_1|)((\frac{k-2}{k-1})^{(k-2)}-1) = \]\[2(|V| - |O|)(1 - (\frac{k-2}{k-1})^{(k-2)}) - 2|M_1|(1 - 2 \cdot (\frac{k-2}{k-1})^{(k-2)}).\]

Next, we show that $(1 - 2 \cdot (\frac{k-2}{k-1})^{(k-2)})\geq 0$.
Let $f(k) = (\frac{k-2}{k-1})^{k-2}$, for $k\geq 3$. Thus, \[f'(k) = \frac{(k-2)^{k-2}\big(\ln(\frac{k-2}{k-1})(k-1) + 1\big)}{(k-1)^{k-1}}.\] 
Now, $\frac{(k-2)^{k-2}}{(k-1)^{k-1}} > 0$.
In addition, it is known that $\ln(x) \leq x-1$~\cite{10.2307/3615890}, and thus $\ln(\frac{k-2}{k-1})(k-1) + 1 \leq -\frac{1}{k-1}(k-1) + 1 = 0$. Therefore, for all $k\geq3$, $f'(k) \leq 0$ and $f(k) \leq f(3) = \frac{1}{2}$.

Overall, since $|M_1| \leq \frac{|V|-|O|}{2}$, 
\[u(P) \geq 2(|V| - |O|)(1 - (\frac{k-2}{k-1})^{(k-2)}) - 2\cdot\frac{|V|-|O|}{2}(1 - 2 \cdot (\frac{k-2}{k-1})^{(k-2)})\] \[ = %
|V|-|O|.\]
Now, since in $Opt$ there are at least $|O|$ singletons, then $u(Opt)$ is at most $(|V|-|O|)\cdot(k-1)$, which occurs when all nodes are partitioned into cliques of size $k$ (except those in $O$). That is, 
\[ u(P) \geq \frac{u(Opt)}{k-1}.\]
\end{proof}

Since finding a maximum matching in a graph can be computed in $O(|E|\sqrt{|V|})$, Algorithm \ref{alg:MnM} runs in $O(kn^{2.5})$ time.

\subsection{Approximation Ratio for Weighted Setting}
We now show that in the weighted setting MnM provides an approximation ratio of $\frac{1}{k}$ for the MaxUtil problem with an odd $k$ and $\frac{1}{k-1}$ for the problem with an even $k$. Specifically, we show that the first step of the algorithm, which finds a maximum weight matching, provides such an approximation ratio.
\begin{theorem}
\label{thm:approx_weight}
    Algorithm \ref{alg:MnM} provides a solution for the MaxUtil problem in the weighted setting with an approximation ratio of $\frac{1}{k}$ for an odd $k$ and an approximation ratio of $\frac{1}{k-1}$ for an even $k$.
\end{theorem}

\begin{figure}[ht]
    \centering
\resizebox{3in}{!}{
    \begin{subfigure}{.45\columnwidth}
    \centering
        \begin{tikzpicture}[scale=0.65,node distance={12.5mm}, main/.style = {draw, circle, scale = 0.8}]
            \foreach \phi in {1,...,7}{
                \node[main] (\phi) at (360/7 * \phi:2.5cm){$v_{\phi}$};
            }
            \draw [ultra thick] (1) -- (2);
            \draw [ultra thick] (3) -- (7);
            \draw [ultra thick] (4) -- (6);    
            \draw [blue, ultra thick] (2) -- (3);
            \draw [blue, ultra thick] (1) -- (4);
            \draw [blue, ultra thick] (5) -- (7);
            \draw [red, ultra thick] (3) -- (4);
            \draw [red, ultra thick] (2) -- (5);
            \draw [red, ultra thick] (1) -- (6);
            \draw [yellow, ultra thick] (4) -- (5);
            \draw [yellow, ultra thick] (3) -- (6);
            \draw [yellow, ultra thick] (2) -- (7);
            \draw [pink, ultra thick] (5) -- (6);
            \draw [pink, ultra thick] (4) -- (7);
            \draw [pink, ultra thick] (3) -- (1);
            \draw [orange, ultra thick] (6) -- (7);
            \draw [orange, ultra thick] (5) -- (1);
            \draw [orange, ultra thick] (4) -- (2);
            \draw [green, ultra thick] (7) -- (1);
            \draw [green, ultra thick] (6) -- (2);
            \draw [green, ultra thick] (5) -- (3);
        \end{tikzpicture}
    \end{subfigure}
    \begin{subfigure}{.45\columnwidth}        
    \centering
        \begin{tikzpicture}[scale=0.55, node distance={12.5mm},  main/.style = {draw, circle, scale = 0.8}]
            \foreach \phi in {1,...,7}{
                \node[main] (\phi) at (360/7 * \phi:3cm){$v_{\phi}$};}
            \node[main] (8) at (0, 0){$v_{8}$};
            \draw [ultra thick] (2) -- (7);
            \draw [ultra thick] (3) -- (6);
            \draw [ultra thick] (4) -- (5);    
            \draw [ultra thick] (1) -- (8);    
            \draw [blue, ultra thick] (3) -- (1);
            \draw [blue, ultra thick] (4) -- (7);
            \draw [blue, ultra thick] (5) -- (6);
            \draw [blue, ultra thick] (2) -- (8);
            \draw [red, ultra thick] (3) -- (8);
            \draw [red, ultra thick] (4) -- (2);
            \draw [red, ultra thick] (5) -- (1);
            \draw [red, ultra thick] (6) -- (7);
            \draw [yellow, ultra thick] (4) -- (8);
            \draw [yellow, ultra thick] (5) -- (3);
            \draw [yellow, ultra thick] (6) -- (2);
            \draw [yellow, ultra thick] (7) -- (1);
            \draw [pink, ultra thick] (5) -- (8);
            \draw [pink, ultra thick] (6) -- (4);
            \draw [pink, ultra thick] (7) -- (3);
            \draw [pink, ultra thick] (1) -- (2);
            \draw [orange, ultra thick] (6) -- (8);
            \draw [orange, ultra thick] (5) -- (7);
            \draw [orange, ultra thick] (4) -- (1);
            \draw [orange, ultra thick] (3) -- (2);
            \draw [green, ultra thick] (7) -- (8);
            \draw [green, ultra thick] (1) -- (6);
            \draw [green, ultra thick] (2) -- (5);
            \draw [green, ultra thick] (3) -- (4);
        \end{tikzpicture}
    \end{subfigure}
    }
    \caption{Examples for edge coloring in graphs with an odd and an even number of vertices.}
    \label{fig:partition}
\end{figure}
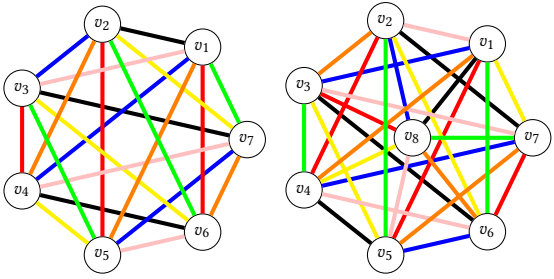

\begin{proof}
let $Opt = \{S_1, S_2, ...\}$ be an optimal $k$-bounded partition, and let $M_i$ be a maximum weight matching for coalition $S_i$. 
Due to Behzad et al. [\citeyear{behzad1967colour}], for a graph with $n$ vertices, there exists a proper edge coloring using $n$ colors for an odd $n$ and $n-1$ colors for an even $n$ (examples for these coloring for $n=7$ and for $n=8$ are shown in Figure \ref{fig:partition}). 
Clearly, each color defines a matching in the graph. Thus, each coalition $S_i$ can be partitioned into $|S_i|$ disjoint matchings, if $|S_i|$ is odd, and $|S_i|-1$ disjoint matchings otherwise. We note that since the matchings are disjoint and cover the entire graph, the sum of the weights of all matchings equals the sum of the weights of all edges in the graph induced by $S_i$.
In addition, for each $i$, $|S_i| \leq k$. Let $M_{S_i}$ be a maximum weight matching for $S_i$, imposed by one of the colors. Now, since a maximum is at least as great as the average, and since $M_i$ is a maximum weight matching for coalition $S_i$, $\sum_{e \in M_i} w(e) \geq \sum_{e \in M_{S_i}} w(e) \geq \frac{\sum_{v\in S_i}{W(v,S_i)}}{k}$ for an odd $k$ and $\sum_{e \in M_i} w(e) \geq \sum_{e \in M_{S_i}} w(e) \geq \frac{\sum_{v\in S_i}{W(v,S_i)}}{k-1}$ for an even $k$.
%
Let $M$ be the maximum weight matching for $G$ found by Algorithm \ref{alg:MnM} in its first step. Clearly, $\sum_{e \in M} w(e) \geq \sum_i \sum_{e \in M_i} w(e)$. In addition, for the $k$-bounded partition $P$ that MnM returns, $u(P) \geq \sum_{e \in M} w(e)$. Therefore, in the weighted setting, Algorithm \ref{alg:MnM} provides a solution for the MaxUtil problem with an approximation ratio of $\frac{1}{k}$ for an odd $k$ and an approximation ratio of $\frac{1}{k-1}$ for an even $k$.

\end{proof}

We refer the reader to the appendix for results related to the tightness of the approximation ratio of MnM.
    
    

\section{Stability}
When considering a stability concept $c$, we analyze the following two problems:
\begin{itemize}
    \item Existence: determine whether for any $(G,k)$ there exists a partition that satisfies $c$.
    \item Finding: given $(G,k)$, decide if there exists a partition that satisfies $c$ and if so, find such a partition.
\end{itemize}

\subsection{Core}
We begin with the unweighted setting. We show that for $k=3$ the core is never empty, and we present Algorithm \ref{alg:core}, a polynomial time algorithm that finds a $3$-bounded partition $P$ in the core. The algorithm begins with all agents in singletons and iteratively considers for each $3$-bounded coalition whether it strongly blocks the current partition.

\begin{algorithm}[ht]
    \caption{Finding a $3$-bounded partition in the core}
    \label{alg:core}
    \SetAlgoLined
    \textbf{Input}:
    A graph $G(V,E)$.\\
    \KwResult{A $3$-bounded partition $P$ of $V$ in the core.} 
    $P \leftarrow \{\{v\}$ for every $v \in V\}$\\
    $V' \leftarrow V$\\
    outerLoop:\\ 
    \label{core:outerLoop}
        \For{$S \subset V'$, such that $|S|=2$ OR $|S|=3$}{
            
            \If{$\forall v \in S, W(v ,S) > u(v,P)$}{
            \label{core:if}
                $P \leftarrow P^{-S}$ \\
                \If{$S$ is clique of size $3$}{
                    $V' \leftarrow V' \setminus S$\\ \label{alg:core:line:remove}
                }
                \textbf{goto} outerLoop
            }
        }
        \textbf{return} P
\end{algorithm}

\begin{theorem}
In the unweighted setting, there always exists a $3$-bounded partition in the core, and it can be found in polynomial time.
\end{theorem}

\begin{proof}
Consider Algorithm \ref{alg:core}.
Note that for every $3$-bounded partition $P$, if $S \in P$ is a clique of size $3$ then every $v \in S$ cannot belong to any strongly blocking coalition. 
Therefore, Algorithm~\ref{alg:core} removes such vertices from $V'$ (in line~\ref{alg:core:line:remove}).
Clearly, if Algorithm \ref{alg:core} terminates, the $3$-bounded partition $P$ is in the core.
We now show that Algorithm \ref{alg:core} must always terminate, and it runs in polynomial time.
The algorithm initiates a new iteration (line  \ref{core:outerLoop}) whenever the if statement in line \ref{core:if} is true, which can happen when the blocking coalition $S$, is one of the following:
\begin{itemize}
    \item Only singletons (i.e., two or three singletons). Then, $u(P)$ increases by at least $2$.
    \item One agent from a coalition in which she has one neighbor, and two singleton agents. Then, $u(P)$ increases by at least $2$.
    \item Two agents, each from a coalition with a single neighbor, and one singleton agent. Then, $S$ must be a clique of size $3$, which increases $u(P)$ by $2$.
    \item Three agents, each from a coalition with a single neighbor. Then, $S$ must also be a clique of size $3$; however, $u(P)$ remains the same.
\end{itemize}
Overall, either $u(P)$ has increased by at least $2$ or $S$ is a clique of size $3$ and thus its vertices are removed from further consideration (in line~\ref{alg:core:line:remove}). 
Since $u(P)$ is bounded by $2|E|$ and the number of vertices is finite, the algorithm must terminate after at most $|E|+|V|/3$ iterations.
\end{proof}

For $k>3$ it is unclear whether the core can be empty, and how to find a partition in the core.
Indeed, we show in simulation that a simple heuristic always finds a partition that is in the core.
Our heuristic function works as follows:
\begin{enumerate}
    \item Start with a $k$-bounded partition $P$, where all the agents are singletons.
    \item Iterate randomly over all the $k$-bounded coalitions until a coalition $S$ is found, which strongly blocks the partition $P$.
    \item Update $P$ to be $P^{-S}$, and return to step $(2)$.
\end{enumerate}
The heuristic terminates when either there is no strongly blocking coalition for the partition $P$ (i.e., $P$ is in the core), or when in $100$ consecutive iterations the algorithm only visits partitions that it has already seen before. In the latter case, we restart the search from the very beginning. 

We test our heuristic function for $k=5$ over more than $100$ million random graphs of different types: 
\begin{itemize}
\item Random graphs of size $30$ with probability of $0.5$ for rewiring each edge.
\item Random trees of size $30$.
\item Random connected Watts–Strogatz small-world graphs of size $30$, where each node is joined with its $5$ nearest neighbors in a ring topology and with a probability of $0.5$ for rewiring each edge.
\end{itemize}

Our heuristic always found a $k$-bounded partition that is in the core. Moreover, we had to restart the heuristic in only $33$ instances, and then a $k$-bounded partition in the core was found.

We continue with the analysis of the weighted setting. We first show that the core may be empty. Specifically, Figure \ref{fig:emptyCore} provides an example of such a graph for $k=3$. We use a computer program that iterates over all possible $3$-bounded partitions, for verifying that each such $3$-bounded partition has a strongly blocking $3$-bounded coalition.

\begin{figure}
\centering
\resizebox{2.2in}{!}{
\begin{tikzpicture}[node distance={14mm}, thick, main/.style = {draw, circle, scale=0.8}]
    
    \node[main] (1) {$v_1$};
    \node[main] (2) [below of=1] {$v_2$}; 
    \node[main] (3) [below left of=2] {$v_3$}; 
    \node[main] (4) [below left of=3] {$v_4$};     
    \node[main] (5) [above left of=4] {$v_5$}; 
    \node[main] (6) [below right of=2] {$v_6$}; 
    \node[main] (7) [below left of=6] {$v_7$}; 
    \node[main] (8) [below right of=6] {$v_8$}; 
    \node[main] (9) [above right of=8] {$v_9$}; 
    
    \tikzset{mystyle/.style={color=black, scale=0.8}} 
    \tikzset{every node/.style={fill=white, scale=0.8}}
    \path   (1)     edge [mystyle]    node  {$6$} (2)
            (2)     edge [mystyle]    node  {$7$} (3)
            (2)     edge [mystyle]    node  {$5$} (6)
            (3)     edge [mystyle]    node  {$4$} (6)
            (3)     edge [mystyle]    node  {$5$} (4)
            (3)     edge [mystyle]    node  {$4$} (7)
            (6)     edge [mystyle]    node  {$4$} (7)
            (6)     edge [mystyle]    node  {$7$} (8)
            (4)     edge [mystyle]    node  {$6$} (5)
            (4)     edge [mystyle]    node  {$7$} (7)
            (7)     edge [mystyle]    node  {$5$} (8)
            (8)     edge [mystyle]    node  {$6$} (9);

\end{tikzpicture}
}
\caption{An example of a weighted graph in which the core is empty, for $k=3$.}
\label{fig:emptyCore}
\end{figure}
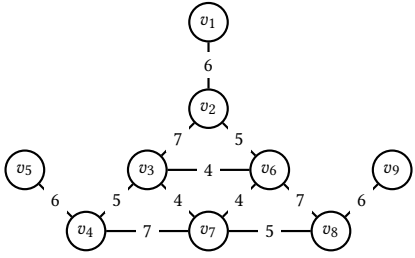

Next, we show that the existence problem for the core in the weighted setting is NP-hard, for any $k \geq 3$. Formally,
\begin{definition}[Core existence problem]
Given a coalition size limit $k$ and a graph $G$, decide whether a $k$-bounded partition in the core exists.
\end{definition}
We reduce from the following problem, which was shown to be strongly NP-complete by \cite{deineko2013two}. 

\begin{definition}[3-Dimensional Stable Roommates with Metric Preferences (\TDSRM)]
 Let $A$ be A set of agents such that $|A| = 3n, n \in \mathbb{N}$, equipped with a metric distance function $d$. Given an agent $i \in A$ and two triples $S_1$ and $S_2$, such that $i \in S_1, S_2$, agent $i$ is said to strictly prefer a triple $S_1$ to $S_2$ if $\sum_{j \in S_1 \setminus \{i\}}d(i,j) < \sum_{j \in S_2 \setminus \{i\}} d(i,j)$. A partition of $A$ into triples is said to be core-stable if there is no triple of agents $T$ in which each of the agents strictly prefers $T$ to her triple in the partition.
 The \TDSRM\ problem asks whether there exists a core-stable partition of $A$ into triples.
\end{definition}

\begin{theorem}\label{thm:core_np_w}
In the weighted setting, the Core existence problem is strongly $NP$-hard, for every $k \geq 3$.
\end{theorem}

\begin{figure}
\centering
\begin{tikzpicture}[node distance={14mm}, thick, main/.style = {draw, circle, scale=0.75}]

\foreach \n / \nn / \x / \y / \c in {
d1/d^1_s/2/2.2/white!80!green,
d2/d^2_s/2/0.8/white!80!green,
ui/u_i/4/1.5/white!80!red,
uj/u_j/5/2.5/white!80!red,
t1/t_1/0.4/2.5/white!80!blue,
t2/t_2/-1/1.5/white!80!blue,
t3/t_3/0.4/0.5/white!80!blue} {
    \node[main, fill=\c] at (\x, \y) (\n) {$\nn$};
}

 \tikzset{mystyle/.style={color=black, scale=0.8, inner sep = 1pt}} 
    \tikzset{every node/.style={fill=white, scale=0.8}}
    
\foreach \s / \e / \w / \d in {
t1/t2/45M/0,
t1/t3/45M/30,
t2/t3/45M/0,
t1/d1/6M/0,
t1/d2/6M/0,
t2/d1/6M/0,
t2/d2/6M/0,
t3/d1/6M/0,
t3/d2/6M/0,
ui/d1/7M/0,
ui/d2/7M/0,
d1/d2/45M/0,
ui/uj/{2M - d(i,j)}/0} {
    \draw (\s) edge [mystyle, bend right =\d]    node  {\small $\w$} (\e);
}

\end{tikzpicture}
\caption{An illustration of the construction in Theorem \ref{thm:core_np_w} for $k=5$, showing the edges between $T$, $D_s$ for some $s \in \{1, \dots, n + 1\}$ and $u_i$ for some $i \in A$. Additionally, we show the edge between $u_i$ and an arbitrary $u_j, j \in A$.}
\label{fig:gadget}
\end{figure}
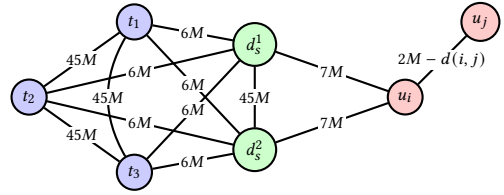

\begin{proof}
Throughout this proof, for a partition $P$ and element $i$, let $P(i)$ denote the coalition of $i$ in $P$.
Let $(A, d)$ be an instance of \TDSRM. We construct an instance $(G = (V, E),k)$ of the Core existence problem. Let $M \coloneqq \max_{i, j \in A}d(i,j) + 1$. We create a set of agent vertices $U \coloneqq \{u_i \mid i \in A\}$. For every pair $\{i, j\} \in A$, create an edge $\{u_i, u_j\}$ with weight $2M - d(i,j)$. Intuitively, these weights enforce that every agent prefers being in a triple to being in a pair or alone. Additionally, restricted to triples, these weights give raise to preferences that are identical to those of $(A, d)$.

If $k = 3$, then $V \coloneqq U$.
If $k \geq 4$, we additionally create $(n + 1)(k - 3)$ many dummy vertices $\Wset \coloneqq \{\Wvertex^t_s \mid s \in \{1, \dots, n + 1\}, t \in \{1, \dots, k - 3\}\}$. The dummy vertices enforce that every agent vertex prefers to be in a triple containing two other agent vertices and $k - 3$ dummies. For every $s \in \{1, \dots, n + 1\}$, the agents $\Wset_s \coloneqq \{\Wvertex^t_s \mid t \in \{1, \dots, k - 3\}\}$ form a clique with edges of weight $\wk$. Thus the dummies want to always be with their clique. There are no other edges between dummy agents. For every $i \in A, s \in \{1, \dots, n + 1\}, t \in \{1, \dots, k - 3\}$, we create an edge $\{u_i, \Wvertex^t_s\}$ with weight $\wa$. We also create additional vertices $T \coloneqq \{t_1, t_2, t_3\}$. There is an edge of weight $\wk$ between every pair of vertices in $T$. For every $\Wvertex^t_s, s \in \{1, \dots, n + 1\}, t \in \{1, \dots, k - 3\}, j \in \{1, 2, 3\}$, we add an edge $\{\Wvertex^t_s, t_j\}$ with weight $\wt$. We set $V \coloneqq U \cup \Wset \cup T$. An illustration of our construction is depicted in Figure~\ref{fig:gadget}.
We begin by showing the following properties of our construction. 

\begin{claim}\label{cl:constr_props}
The following properties hold:
\begin{enumerate}
\item If $k \geq 4$, then for every $s \in \{1, \dots, n + 1\}, t \in \{1, \dots, k - 3\}$, the vertex $\Wvertex^t_s$ strictly prefers a coalition containing $\Wset_s$ to one that does not.\label{cl:constr_propsB} 
\item If $k \geq 4$, then for every $j \in \{1, 2, 3\}$, the vertex $t_j$ strictly prefers a coalition containing $T$ to one that does not. Moreover, for every $s \in \{1, \dots, n + 1\}$, there is no coalition that $t_j$ strictly prefers over $\Wset_s \cup T$. \label{cl:constr_propsT} 
\item Let $U^3, U^2 \subseteq U$ be arbitrary subsets such that $|U^3| = 3$ and $|U^2| = 2$.  If $k \geq 4$, then for every $s \in \{1, \dots, n + 1\}, t \in \{1, \dots, k - 3\}$, the vertex $\Wvertex^t_s$ strictly prefers $\Wset_s \cup U^3$ to $\Wset_s \cup T$, and $\Wset_s \cup T$ to $\Wset_s \cup U^2$.\label{cl:constr_propsBA}
\item If $k \geq 4$, then for every $i \in A$, the vertex $u_i$ strictly prefers a coalition $\coalS$ to $\coalS'$ if $\Wset_s \subseteq S$ for some $s \in \{1, \dots, n + 1\}$ and $\Wset \cap S' = \emptyset$.\label{cl:constr_propsAB}
\item If $k \geq 4$, then for every $i \in A, \{j, \ell\}, \{j', \ell'\} \subseteq A \setminus \{i\}$, $s,s' \in \{1, \dots, n+1\}$ the vertex $u_i$ (strictly) prefers the coalition $\Wset_s \cup \{u_i, u_j, u_{\ell}\}$ to $\Wset_{s'} \cup \{u_i, u_{j'}, u_{\ell'}\}$ if and only if $i$ (strictly) prefers $\{i,j,\ell\}$ to $\{i,j',\ell'\}$.

If $k = 3$, then for every $i \in A, \{j, \ell\}, \{j', \ell'\} \subseteq A \setminus \{i\}$, the vertex $u_i$ (strictly) prefers the coalition $\{u_i, u_j, u_{\ell}\}$ to $\{u_i, u_{j'}, u_{\ell'}\}$ if and only if $i$ (strictly) prefers $\{i,j,\ell\}$ to $\{i,j',\ell'\}$.\label{cl:constr_propsA}
\item If $k \geq 4$,  then for every $i \in A, j, \ell, j' \in A \setminus \{i\}, j \neq \ell$, $s, s' \in \{1, \dots, n + 1\}$, the vertex $u_i$ strictly prefers the $\Wset_s \cup \{u_i, u_j, u_{\ell}\}$ to $\Wset_{s'} \cup \{u_i, u_{j'}\}$ and to $\Wset_{s'} \cup \{u_i\}$.

If $k = 3$, then for every $i \in A, j, \ell, j' \in A \setminus \{i\}, j \neq \ell$, the vertex $u_i$ strictly prefers the coalition $ \{u_i, u_j, u_{\ell}\}$ to $\{u_i, u_{j'}\}$ and to~$\{u_i\}$. \label{cl:constr_propsAn}
\item If $k \geq 4$, then for every $\{x_1, x_2, x_3\}, \{x_4, x_5\}, \{x_6\} \subseteq U \cup T, s \in \{1, \dots, n + 1\}, t \in \{1, \dots, k - 3\}$, the vertex $\Wvertex^t_s$ strictly prefers $\Wset_s \cup \{x_1, x_2, x_3\}$ to $\Wset_s \cup \{x_4, x_5\}$ to $\Wset_s \cup \{x_6\}$ to~$\Wset_s$. \label{cl:constr_propsBn}
\end{enumerate}
\end{claim}
\begin{proof} 

\textbf{Case \eqref{cl:constr_propsB}:} If $k = 4$ the statement is trivially true, because $\Wset_s = \{\Wvertex^1_s\}$ for every $s \in \{1, \dots, n + 1\}$. Assume $k \geq 5$. Let $s \in \{1, \dots, n + 1\}, t \in \{1, \dots, k - 3\}$. Let $\coalS \subseteq V$ be an arbitrary set such that $,|S| \leq k, \Wvertex^t_s \in \coalS$ and $\Wset_s \nsubseteq \coalS$. Then $\uCS{\Wvertex^t_s}{\coalS} \leq (k - 5)\wk + 4 \cdot \wa$. However, $\uCS{\Wvertex^t_s}{\Wset_s} = (k - 4)\wk  = (k - 5)\wk + \wk \geq (k - 5)\wk + 45M> \uCS{\Wvertex^t_s}{\coalS}$.\\

\textbf{Case \eqref{cl:constr_propsT}:}
Let $j \in \{1, \dots, k - 3\}$. Let $\coalS \subseteq V$ be an arbitrary set such that $|\coalS| \leq k, t_j \in \coalS$ and $T \nsubseteq \coalS$. Then $\uCS{t_j}{\coalS} \leq \wk + (k - 2)\wt$. However, $\uCS{t_j}{T} = 2\wk = \wk + \wk> \uCS{t_j}{\coalS}$.\\

\textbf{Case \eqref{cl:constr_propsBA}:}
We have that $\uCS{\Wvertex^t_s}{\Wset_s \cup U^3} = (k - 4)\wk + 3 \cdot \wa = (k - 4) \wk + 21M$, $\uCS{\Wvertex^t_s}{\Wset_s \cup T} = (k - 4)\wk + 3 \cdot \wt = (k - 4) \wk + 18M$ and $\uCS{\Wvertex^t_s}{\Wset_s \cup U^2} = (k - 4)\wk + 2 \cdot \wa = (k - 4)\wk + 14M$. Thus $\uCS{\Wvertex^t_s}{\Wset_s \cup U^3} > \uCS{\Wvertex^t_s}{\Wset_s \cup T} > \uCS{\Wvertex^t_s}{\Wset_s \cup U^2}$.\\

\textbf{Case \eqref{cl:constr_propsAB}:}
We have that $\uCS{u_i}{\coalS} \geq (k - 3)\wa \geq (k - 3)2M + 5M$ and $\uCS{u_i}{\coalS'} \leq (k - 1)2M = (k - 3)2M + 4M$. Thus $\uCS{u_i}{\coalS} > \uCS{u_i}{\coalS'}$.\\

\textbf{Case \eqref{cl:constr_propsA}:}
First assume $k \geq 4$. Then
\begin{align*}
&\uCS{u_i}{\Wset_s \cup \{u_i, u_j, u_{\ell}\}} > \uCS{u_i}{\Wset_{s'} \cup \{u_i, u_{j'}, u_{\ell'}\}} \iff\\
&2M - d(i, j) + 2M - d(i, \ell) + 3 \cdot \wa >\\ &\qquad 2M - d(i, j') + 2M - d(i, \ell') + 3 \cdot \wa \iff\\
&d(i,j) + d(i, \ell) < d(i,j') + d(i, \ell'),
\end{align*}
i.e., $u_i$ strictly prefers $\Wset_s \cup \{u_i, u_j, u_{\ell}\}$ to $\Wset_{s'} \cup \{u_i, u_{j'}, u_{\ell'}\}$ if and only if $i$ strictly prefers $\{i, j, \ell\}$ to $\{i, j', \ell'\}$. The computations for weakly preferring are identical. For the case where $k = 3$, 

\begin{align*}
\uCS{u_i}{\{u_i, u_j, u_{\ell}\}} &> \uCS{u_i}{\{u_i, u_{j'}, u_{\ell'}\}} \iff\\
2M - d(i, j) + 2M - d(i, \ell) &> 2M - d(i, j') + 2M - d(i, \ell') \iff\\
d(i,j) + d(i, \ell) &< d(i,j') + d(i, \ell'),
\end{align*}
instead.\\

\textbf{Case \eqref{cl:constr_propsAn}:}
First assume $k \geq 4$. Then $\uCS{u_i}{\Wset_s \cup \{u_i, u_j, u_{\ell}\}} \geq (k - 3) \wa + 2(M + 1) = (k - 3) \wa + 2M + 2$, $\uCS{u_i}{\Wset_s \cup \{u_i, u_{j'}\}} \leq (k - 3) \wa + 2M$ and $\uCS{u_i}{\Wset_s \cup \{u_i\}} = (k - 3) \wa$. The statement follows.

If $k = 3$, we have instead that $\uCS{u_i}{\{u_i, u_j, u_{\ell}\}} \geq  2(M + 1)$, $\uCS{u_i}{\Wset_s \cup \{u_i, u_{j'}\}} \leq  2M$ and $\uCS{u_i}{ \Wset_s \cup \{u_i\}} = 0$. The statement follows similarly.\\

\textbf{Case \eqref{cl:constr_propsBn}:}
We have that $\uCS{\Wvertex^t_s}{\Wset_s \cup \{x_1, x_2, x_3\}} \geq (k - 4)\wk + 3 \cdot \wt > (k - 4)\wk + 2 \cdot \wa \geq \uCS{\Wvertex^t_s}{\Wset_s \cup \{x_4, x_5\}}$. Similarly we have that $\uCS{\Wvertex^t_s}{\Wset_s \cup \{x_4, x_5\}} \geq (k - 4) \wk + 2 \cdot \wt > (k - 4) \wk + \wa \geq \uCS{\Wvertex^t_s}{\Wset_s \cup \{x_6\}}$. Finally $\uCS{\Wvertex^t_s}{\Wset_s \cup \{x_6\}} \geq (k - 4)\wk + \wt > (k - 4) \wk = \uCS{\Wvertex^t_s}{\Wset_s}$.

\end{proof}

\paragraph{If $(A, d)$ admits a core-stable partition, then $(G = (V, E),k)$ admits a core-stable partition.}

Let $\partitionS = \{\coalS_1, \dots, \coalS_n\}$ be a core stable partition of $A$.  We construct a partition $\partitionS'$ of $V$ as $\partitionS' \coloneqq \{\{u_i \mid i \in \coalS_j \}\cup \Wset_j \mid j \in \{1, \dots, n\}\} \cup \{\Wset_{n + 1} \cup T\}$ if $k \geq 4$ and $\partitionS' \coloneqq \{\{u_i \mid i \in \coalS_j \mid j \in \{1, \dots, n\} \}$ if $k = 3$.
Assume that $\partitionS'$ is not core-stable in $(G, k)$. Then exists $\coalS \subseteq V, |\coalS| \leq k$ that blocks~$\partitionS'$.\\

We start by showing that if $k \geq 4$, any blocking coalition $\coalS$ can have at most $3$ elements from $U$, and the rest of the elements must be from $\Wset$. 

First note that for every $\Wvertex^t_s, s \in \{1, \dots, n\}, t \in \{1, \dots, k - 3\}$, we have that $\uS{\Wvertex^t_s}{ \partitionS'} = (k - 4)\wk + 3\cdot \wa = (k - 4)\wk + 21M$. For every $\Wvertex^t_{n+1}, t \in \{1, \dots, k - 3\}$, we have that $\uS{\Wvertex^t_{n+1}}{ \partitionS'} = (k - 4)\wk + 3 \cdot \wt = (k - 4)\wk + 18M$.
By Claim \ref{cl:constr_props}\eqref{cl:constr_propsT}, no vertex in $T$ wants to deviate.

If $\Wvertex^t_s \in \coalS$, where $ s \in \{1, \dots, n + 1\}, t \in \{1, \dots, k - 3\}$, then $\Wset_s \subseteq \coalS$. Assume, towards a contradiction, that there is some dummy agent $\Wvertex^t_s \in \coalS, s \in \{1, \dots, n + 1\}, t \in \{1, \dots, k - 3\}$ such that $\Wset_s \nsubseteq \coalS$. By Claim \ref{cl:constr_props}\eqref{cl:constr_propsB}, the vertex $\Wvertex^t_s$ prefers $\Wset_s \subseteq \partitionS'(\Wvertex^t_s)$ to $\coalS$, a contradiction to $\coalS$ blocking. 

Moreover, if $\Wset_s \subseteq \coalS$, then $\Wset_{s'} \nsubseteq \coalS$ for any $s' \in \{1, \dots, n + 1\} \setminus \{s\}$. Assume, towards a contradiction, that $\Wset_s \cup \Wset_{s'} \subseteq \coalS$ for some $s,s' \in \{1, \dots, n + 1\}, s \neq s'$. Note that this is only possible if $2 (k - 3) \leq k \iff k \leq 6$. Then for every $\Wvertex^t_s, t \in \{1, \dots, k - 3\}$, $u(\Wvertex^t_s, \coalS) \leq (k - 4)\wk + (6 - k)\wt$. Since $6 - k \leq 3$, it follows that $\uCPS{\Wvertex^t_s}{\coalS}{\partitionS'} \leq \uS{\Wvertex^j_i}{ \partitionS'}$, a contradiction. Therefore, if $\coalS$ contains some element in $\Wvertex^t_s \in \Wset$, then $\Wset \cap \coalS = \Wset_s$.

If $\Wset_s \nsubseteq (\Wset \cap \coalS) $ for any $s \in \{1, \dots, n + 1\}$, then no agent in $U \cap \coalS$ wants to deviate by Claim \ref{cl:constr_props}\eqref{cl:constr_propsAB}, a contradiction. Thus there must be some $s \in \{1, \dots, n + 1\}$ such that $\Wset_s \subseteq \coalS$. If $|\coalS \cap U| >3$, then $|\coalS \cap \Wset| < k - 3$, and thus $\Wset_s \nsubseteq \coalS$ for any $s \in \{1, \dots, n + 1\}$, contradicting the previous statement. 

Now we continue in the generic case where $k \geq 3$.
If $|\coalS \cap U| < 3$, then by Claim \ref{cl:constr_props}\eqref{cl:constr_propsAn}, every $u_i \in \coalS, i \in A$ prefers $\partitionS'(u_i)$ to $\coalS$, a contradiction to $\coalS$ blocking.
Thus it must be that $|\coalS \cap U| = 3$. By Claim \ref{cl:constr_props}\eqref{cl:constr_propsA}, every $u_i \in \coalS, i \in A$ strictly prefers $\partitionS'(u_i)$ to $\coalS$ if and only if $i$ prefers $\hat{\coalS} \coloneqq \{i \in A\mid u_i \in \coalS \cap U\}$ to $\partitionS(i)$. But then $\hat{\coalS}$ blocks $\partitionS$, a contradiction.

\paragraph{If $(G,k)$ admits a core-stable partition, then $(A, d)$ admits a core-stable partition.}

Let $\partitionS'$ be a core stable partition of $V$.

We first show that if $k \geq 4$, then every coalition in $\partitionS'$ contains a clique of the vertices in $\Wset$. Assume, towards a contradiction, that there is $\Wvertex^t_s, s \in \{1, \dots, n + 1\}, t \in \{1, \dots, k - 3\}$ such that $\Wset'_s \nsubseteq \partitionS'(\Wvertex^t_s)$. By Claim \ref{cl:constr_props}\eqref{cl:constr_propsB}, $\Wvertex^{t}_s$ prefers $\Wset_s$ to $\partitionS'(\Wvertex^j_i)$. Thus $\Wset_s$ blocks $\partitionS'$, a contradiction. By identical reasoning on Claim \ref{cl:constr_props}\eqref{cl:constr_propsT}, we have that $\partitionS'(t_1) = \partitionS'(t_2) = \partitionS'(t_3)$. For every $s \in \{1, \dots, n + 1\}, t \in \{1, \dots, k - 3\}$, let $\partitionS'(\Wset_s) \coloneqq \partitionS'(\Wvertex^t_s)$. Because all the elements in $\Wset_s$ are in the same coalition in $\partitionS'$, this is well-defined. Similarly let $\partitionS'(T) \coloneqq \partitionS'(t_1)$.

Next we show that for every $s, s' \in \{1, \dots, n + 1\}, s \neq s'$, we have that $\partitionS'(\Wset_s) \neq \partitionS'(\Wset_{s'})$. Assume, towards a contradiction, that $\partitionS'(\Wset_s) = \partitionS'(\Wset_{s'})$ for some $s, s' \in \{1, \dots, n + 1\}, s \neq s'$. Then there are at most $n$ coalitions that contain $\Wset_{s''}$ for some $s'' \in \{1, \dots, n + 1\}$. Since $|U \cup T | = 3n + 3$, there must be at least three vertices in $U \cup T$ that do not have any vertices from $\Wset$ in their coalition. Since $\partitionS'(t_1) = \partitionS'(t_2) = \partitionS'(t_3)$, if one of the vertices in $T$ does not have vertices from $\Wset$ in their coalition, none of them does. In this case $\Wset_{s} \cup T$ blocks $\partitionS'$ by Claim \ref{cl:constr_props}\eqref{cl:constr_propsT} and \eqref{cl:constr_propsBn}. In the case that none of these vertices is in $T$, let us call an arbitrary size-3 subset of them $U^3$. The coalition $U^3 \cup \Wset_{s}$ blocks $\partitionS'$ by Claim \ref{cl:constr_props}\eqref{cl:constr_propsAB} and \eqref{cl:constr_propsBn}.

Next we show that if $k \geq 4$, then every coalition containing some triple of vertices in $U$ must also contain a clique of vertices from $\Wset$. Assume, towards a contradiction, that there is some $i \in A$ such that $\Wset_s \nsubseteq \partitionS'(u_i)$ for any $s \in \{1, \dots, n + 1\}$ and $|\partitionS'(u_i) \cap U| \geq 3$. Let $U^3$ be an arbitrary size $3$ subset of $\partitionS'(u_i) \cap U$. Since there are $3n$ agents in $U$ and $n + 1$ cliques in $\Wset$, there must be a clique $\Wset_s, s \in \{1, \dots, n + 1\}$ such that $|\partitionS'(\Wset_s) \cap U| \leq 3$. By Claim \ref{cl:constr_props}\eqref{cl:constr_propsAB}, every agent $u_j \in U^3$ prefers $U^3 \cup \Wset_s$ to $\partitionS'(u_j)$. By Claim \ref{cl:constr_props}\eqref{cl:constr_propsBA} and \eqref{cl:constr_propsBn}, every agent in $\Wset_s$ prefers $U^3 \cup \Wset_s$ to $\partitionS'(\Wset_s)$. Thus $\partitionS'$ is not stable.

We proceed to show that for every $u_i, i \in A$, $|\partitionS'(u_i) \cap U| = 3$. Assume, towards a contradiction, that for some $u_i, i \in A$, $|\partitionS'(u_i) \cap U| > 3$. If $k = 3$, this trivially leads to a contradiction. Thus assume $k \geq 4$. Then $|\partitionS'(u_i) \cap \Wset| < k - 3$. Since every vertex in $\Wset$ must have its whole clique in the coalition, $\partitionS'(u_i) \cap \Wset = \emptyset$. Since there are $\{1, \dots, n + 1\}$ cliques in $\Wset$ and $3n$ agents in $U$, there must be some $\Wset_{s'}, s' \in \{1, \dots, n + 1\}$ such that $|\partitionS'(\Wset_{s'}) \cap U| < 3$. Let $U^3$ be an arbitrary subset of $\partitionS'(u_i)$ such that $|U^3| = 3$. By Claim \ref{cl:constr_props}\eqref{cl:constr_propsA}, every vertex in $U^3$ prefers $U^3 \cup \Wset_{s'}$ to its current coalition. By Claim \ref{cl:constr_props}\eqref{cl:constr_propsBA} and \eqref{cl:constr_propsA}, the vertices in $\Wset_{s'}$ also prefer $U^3 \cup \Wset_{s'}$ to $\partitionS'(\Wset_{s'})$, meaning that $U^3 \cup \Wset_{s'}$ blocks $\partitionS'$, a contradiction.

If there is some $u_i, i \in A$, such that $|\partitionS'(u_i) \cap U| < 3$, then by previous paragraph there must be at least two other agents in $j, \ell \in A$ such that $|\partitionS'(u_j) \cap U| < 3$ and $|\partitionS'(u_{\ell}) \cap U| < 3$. If $k \geq 4$, since there are $3n$ agents in $U$ and $n + 1$ cliques in $\Wset$, there must be some clique $\Wset_s, s \in \{1, \dots, n + 1\}$ such that $|\partitionS'(\Wset_s) \cap U| \leq 3$. By Claim \ref{cl:constr_props}\eqref{cl:constr_propsB} and \eqref{cl:constr_propsAn}, every agent $x \in \{u_i, u_j, u_{\ell}\}$ prefers $\{u_i, u_j, u_{\ell}\} \cup \Wset_s$ to $\partitionS'(x)$. By Claim \ref{cl:constr_props}\eqref{cl:constr_propsBA} and \eqref{cl:constr_propsBn}, every agent in $\Wset_s$ prefers $\{u_i, u_j, u_{\ell}\} \cup \Wset_s$ to $\partitionS'(\Wset_s)$. Thus $\partitionS'$ is not stable. If $k = 3$, then by Claim \ref{cl:constr_props}\eqref{cl:constr_propsAn}, every agent $x \in \{u_i, u_j, u_{\ell}\}$ prefers $\{u_i, u_j, u_{\ell}\}$ to $\partitionS'(x)$, a contradiction.

Therefore, if $k \geq 4$, then every coalition in $\partitionS'$ containing vertices in $U$ must be of the form $U^3 \cup \Wset_s$, where $s \in \{1, \dots, n + 1\}$ and $U^3 \subseteq U, |U^3| = 3$. If $k = 3$, every coalition $\coalS \in \partitionS'$ must satisfy $|\coalS| = 3$. We construct a partition $\partitionS$ of $A$ as follows: For every $i \in A$, we set $\partitionS(i) = \{j \mid u_j \in U \cap \partitionS'(u_i)\}$. By previous reasoning, this must partition $A$ into triples.

Assume, towards a contradiction, that $\partitionS$ is not stable. Then there must be some triple $\coalS = \{i,j,\ell\}$ such that $\coalS$ blocks $\partitionS$. Let $\coalS' \coloneqq \{u_i, u_j, u_{\ell}\}$. 
If $k \geq 4$, we observe that since there are $3n$ agents in $U$ and $n + 1$ cliques in $B$, there must be some clique $\Wset_s, s \in \{1, \dots, n + 1\}$ such that $\partitionS(\Wset_s) \cap U = \emptyset$. By Claim \ref{cl:constr_props}\eqref{cl:constr_propsBn}, every agent in $\Wset_s$ prefers $\coalS' \cup \Wset_s$ to $\partitionS'(\Wset_s)$. By Claim \ref{cl:constr_props}\eqref{cl:constr_propsA}, every agent $u_x \in \coalS'$ strictly prefers $\coalS' \cup \Wset_s$ to $\partitionS'(u_x)$. Thus $\partitionS'$ is not stable, a contradiction. If $k = 3$, by Claim \ref{cl:constr_props}\eqref{cl:constr_propsA}, every agent $u_x \in \coalS'$ strictly prefers $\coalS'$ to $\partitionS'(u_x)$. Thus $\partitionS'$ is not stable, a contradiction.
\end{proof}

\subsection{Strict Core (SC)}
We first show that for every size limit, $k$, there is at least one graph where there is no $k$-bounded partition in the strict core.
Indeed, given a size limit $k$, we build the graph $G(V,E)$, which is a clique of size $k+1$. For every partition $P$ of $V$, let $S$ be a coalition in $P$ such that $|S| < k$. Now, any set of agents of size $k$ that also contains some $v \in S$ is a weakly blocking $k$-bounded coalition for $P$. 
%
Furthermore, even verifying the existence of the strict core is a hard problem.

\begin{definition}[$SC$ existence problem]
Given a coalition size limit $k$ and a graph $G$, decide whether a $k$-bounded partition in the strict core exists.
\end{definition}

For the hardness proof, we define for each $k \in \mathbb{N}$ the $Cliques_k$ problem, which is as follows.
\begin{definition}[$Cliques_k$]
Given an undirected and unweighted graph $G(V,E)$, decide whether $V$ can be partitioned into disjoint cliques, such that each clique is composed of exactly $k$ vertices.
\end{definition}
Clearly, $Cliques_2$ can be decided in polynomial time by computing a maximum matching of the graph $G$, $M$, and testing whether $|M| = \frac{|V|}{2}$.
However, $Cliques_k$ becomes hard when $k\geq 3$.
\begin{lemma}
\label{lem:1}
$Cliques_k$ is $NP$-Complete for every $k\geq3$.
\end{lemma}
\begin{proof}
Clearly, $Cliques_k$ is $NP$ for every $k$.
We use induction to show that any $Cliques_k$ is $NP$-Hard for every $k\geq 3$.
$Cliques_3$ is known as the `partition into triangles' problem, which was shown to be $NP$-Complete \cite{garey1979computers}.
Given that $Cliques_k$ is $NP$-Hard we show that $Cliques_{k+1}$ is also $NP$-Hard.
Given an instance of the $Cliques_k$ on a graph $G(V,E)$, we construct the following instance. We build a graph $G'(V',E')$, in which we add a set of nodes $\hat{V} = {\hat{v}_1, ..., \hat{v}_{\frac{|V|}{k}}}$, i.e., $V' = V \cup \hat{V}$. If $e \in E$ then $e \in E'$, and for every $v \in V, \hat{v} \in \hat{V}$ we add $(v,\hat{v})$ to $E'$.
Clearly, $V$ can be partitioned into disjoint cliques with exactly $k$ vertices if and only if $V'$ can be partitioned into disjoint cliques with  exactly $k+1$ vertices.
\end{proof}

\begin{theorem}
\label{theorem:10}
The $SC$ existence problem is $NP$-hard.
\end{theorem}
\begin{proof}
Given an instance of the $Cliques_k$ on a graph $G(V,E)$, we construct the following instance. We build a graph $G'(V',E')$ such that $V'$ contains all the nodes from $V$. In addition, for every $v_x \in V$ we add the nodes $\hat{v}_x$ and $v_x^1,\ldots,v_x^{k-1}$ to $V'$. Now, $E'$ contains all the edges of $E$, and for every $v_x \in V$ and $1 \leq i \leq k-1$ we add $(v_x,v_x^i), (v_x^i, \hat{v}_x)$ to $E'$. Finally, for every $v_x \in V$ and $1 \leq i,j \leq k-1$, $i \neq j$ we add $(v_x^i, v_x^j)$ to $E'$. 
We first show that if $G$ cannot be partitioned into disjoint cliques of size $k$, then the strict core is empty. Indeed, assume that $G$ cannot be partitioned into disjoint cliques of size $k$, and let $P$ be a $k$-bounded partition of $V'$. Then, there is at least one vertex $v_x \in V$ that belongs to a coalition $S\in P$, such that either:
    (1) $W(v_x,S) < k-1$, or
    (2) $v_x^i \in S$ for some $i$ between $1$ and $k-1$.
In case $1$, the coalition $\{v_x, v_x^1, \ldots, v_x^{k-1}\}$ is a weakly blocking $k$-bounded coalition. In case $2$, the coalition $\{\hat{v}_x, v_x^1, \ldots, v_x^{k-1}\}$ is a weakly blocking $k$-bounded coalition.
Therefore, if the strict core is not empty, then $G$ can be partitioned into disjoint cliques of size $k$.

We now show that if $G$ can be partitioned into disjoint cliques of size $k$, then the strict core is not empty. Clearly, in this case $G'$ can be partitioned into disjoint cliques of size $k$, and this partition is in the strict core. Therefore, if the strict core is empty, then $G$ cannot be partitioned into disjoint cliques of size $k$.
\end{proof}

    
    

\subsection{Contractual Strict Core (CSC)}
We show that the CSC is never empty. Indeed, given any $(G,k)$, the following algorithm finds a $k$-bounded partition in the CSC:
\begin{enumerate}
    \item Start with a $k$-bounded partition $P$, where all the agents are singletons.
    \item Iterate over all the coalitions in $P$ until two coalitions, $S_1, S_2$, are found, such that $|S_1|+|S_2| \leq k$ and $u(P) < u(P^{-S_1 \cup S_2})$.
    \item Update $P$ to be $P^{-S_1 \cup S_2}$, and return to step $(2)$.
\end{enumerate}
The algorithm terminates when step $2$ does not find two coalitions that meet the required conditions.

 \begin{theorem}
There always exists a $k$-bounded partition in the CSC, and it can be found in polynomial time.
\end{theorem}

\begin{proof}
At each iteration, the number of the coalitions in $P$ decreases and thus the algorithm must terminate after at most $k-1$ iterations.
Consider the $k$-bounded partition $P$ when the algorithm terminates. Clearly, there are no two coalitions in $P$ that can benefit from breaking off and joining together.
In addition, observe that every coalition $S \in P$ is a connected component. Thus, no coalition $S' \subsetneq S$ can break off without decreasing the utility of at least one agent from $S\setminus S'$. 
Therefore, $P$ is in the CSC.   
\end{proof}

\section{Conclusions and Future Work}
In this paper, we study ASHGs with a bounded coalition size. We provide MnM, an approximation algorithm for maximizing the utilitarian social welfare and study the computational aspects of 
the core, the SC, and the CSC. 
We note that MnM can be improved by running the algorithm and iteratively joining together any two coalitions that improve the social welfare (without violating the size constraint). This improved version is guaranteed to find a partition that is in the CSC while maintaining the approximation ratio for the MaxUtil problem. Unfortunately, this improvement does not result in an improved approximation ratio when $k=3$, and whether it improves the approximation ratio when $k>3$ remains an open problem. Generally, providing an inapproximability result or a better approximation algorithm for the MaxUtil problem is an important open problem.
Furthermore, the existence of the core in the unweighted setting when $k>3$ is an essential open problem.

In addition, there are several interesting directions for extending our work. 
Since the MaxUtil problem is computationally hard,  it will be interesting to investigate some variants. For example, the problem of finding a $k$-bounded partition, such that each agent will be matched with at least one friend in its coalition.
%
Another interesting research direction is to incorporate skills in our model, motivated by coalitional skill games \cite{bachrach2008coalitional}. That is, each agent has a set of skills, and each coalition is required to have at least one agent that acquires each of the skills.
\begin{acks}
This research has been partly supported by the Ministry of Science, Technology \& Space (MOST), Israel.
\end{acks}

\bibliographystyle{ACM-Reference-Format}
\bibliography{ridesharing}


\begin{thebibliography}{25}


\ifx \showCODEN    \undefined \def \showCODEN     #1{\unskip}     \fi
\ifx \showDOI      \undefined \def \showDOI       #1{#1}\fi
\ifx \showISBNx    \undefined \def \showISBNx     #1{\unskip}     \fi
\ifx \showISBNxiii \undefined \def \showISBNxiii  #1{\unskip}     \fi
\ifx \showISSN     \undefined \def \showISSN      #1{\unskip}     \fi
\ifx \showLCCN     \undefined \def \showLCCN      #1{\unskip}     \fi
\ifx \shownote     \undefined \def \shownote      #1{#1}          \fi
\ifx \showarticletitle \undefined \def \showarticletitle #1{#1}   \fi
\ifx \showURL      \undefined \def \showURL       {\relax}        \fi
\providecommand\bibfield[2]{#2}
\providecommand\bibinfo[2]{#2}
\providecommand\natexlab[1]{#1}
\providecommand\showeprint[2][]{arXiv:#2}

\bibitem[\protect\citeauthoryear{Aziz, Brandt, and Seedig}{Aziz et~al\mbox{.}}{2013}]%
        {aziz2013computing}
\bibfield{author}{\bibinfo{person}{Haris Aziz}, \bibinfo{person}{Felix Brandt}, {and} \bibinfo{person}{Hans~Georg Seedig}.} \bibinfo{year}{2013}\natexlab{}.
\newblock \showarticletitle{Computing Desirable Partitions in Additively Separable Hedonic Games}.
\newblock \bibinfo{journal}{\emph{Artificial Intelligence}}  \bibinfo{volume}{195} (\bibinfo{year}{2013}), \bibinfo{pages}{316--334}.
\newblock


\bibitem[\protect\citeauthoryear{Aziz and Savani}{Aziz and Savani}{2016}]%
        {aziz2016hedonic}
\bibfield{author}{\bibinfo{person}{Haris Aziz} {and} \bibinfo{person}{Rahul Savani}.} \bibinfo{year}{2016}\natexlab{}.
\newblock \showarticletitle{Hedonic Games}.
\newblock In \bibinfo{booktitle}{\emph{Handbook of Computational Social Choice}}, \bibfield{editor}{\bibinfo{person}{Felix Brandt}, \bibinfo{person}{Vincent Conitzer}, \bibinfo{person}{Ulle Endriss}, \bibinfo{person}{J{\'{e}}r{\^{o}}me Lang}, {and} \bibinfo{person}{Ariel~D. Procaccia}} (Eds.). \bibinfo{publisher}{Cambridge University Press}, Chapter~15, \bibinfo{pages}{356--376}.
\newblock


\bibitem[\protect\citeauthoryear{Bachrach, Kohli, Kolmogorov, and Zadimoghaddam}{Bachrach et~al\mbox{.}}{2013}]%
        {Bachrach2013}
\bibfield{author}{\bibinfo{person}{Yoram Bachrach}, \bibinfo{person}{Pushmeet Kohli}, \bibinfo{person}{Vladimir Kolmogorov}, {and} \bibinfo{person}{Morteza Zadimoghaddam}.} \bibinfo{year}{2013}\natexlab{}.
\newblock \showarticletitle{Optimal Coalition Structure Generation in Cooperative Graph Games}. In \bibinfo{booktitle}{\emph{Proceedings of the 27th AAAI Conference on Artificial Intelligence}}. \bibinfo{pages}{81--87}.
\newblock


\bibitem[\protect\citeauthoryear{Bachrach and Rosenschein}{Bachrach and Rosenschein}{2008}]%
        {bachrach2008coalitional}
\bibfield{author}{\bibinfo{person}{Yoram Bachrach} {and} \bibinfo{person}{Jeffrey~S Rosenschein}.} \bibinfo{year}{2008}\natexlab{}.
\newblock \showarticletitle{Coalitional Skill Games}. In \bibinfo{booktitle}{\emph{Proceedings of the 7th International Joint Conference on Autonomous Agents and Multiagent Systems}}. \bibinfo{pages}{1023--1030}.
\newblock


\bibitem[\protect\citeauthoryear{Ballester}{Ballester}{2004}]%
        {ballester2004np}
\bibfield{author}{\bibinfo{person}{Coralio Ballester}.} \bibinfo{year}{2004}\natexlab{}.
\newblock \showarticletitle{NP-completeness in Hedonic Games}.
\newblock \bibinfo{journal}{\emph{Games and Economic Behavior}} \bibinfo{volume}{49}, \bibinfo{number}{1} (\bibinfo{year}{2004}), \bibinfo{pages}{1--30}.
\newblock


\bibitem[\protect\citeauthoryear{Banerjee, Konishi, and S{\"o}nmez}{Banerjee et~al\mbox{.}}{2001}]%
        {banerjee2001core}
\bibfield{author}{\bibinfo{person}{Suryapratim Banerjee}, \bibinfo{person}{Hideo Konishi}, {and} \bibinfo{person}{Tayfun S{\"o}nmez}.} \bibinfo{year}{2001}\natexlab{}.
\newblock \showarticletitle{Core in a Simple Coalition Formation Game}.
\newblock \bibinfo{journal}{\emph{Social Choice and Welfare}} \bibinfo{volume}{18}, \bibinfo{number}{1} (\bibinfo{year}{2001}), \bibinfo{pages}{135--153}.
\newblock


\bibitem[\protect\citeauthoryear{Behzad, Chartrand, and Cooper~JR}{Behzad et~al\mbox{.}}{1967}]%
        {behzad1967colour}
\bibfield{author}{\bibinfo{person}{Mehdi Behzad}, \bibinfo{person}{Gary Chartrand}, {and} \bibinfo{person}{John~K Cooper~JR}.} \bibinfo{year}{1967}\natexlab{}.
\newblock \showarticletitle{The Colour Numbers of Complete Graphs}.
\newblock \bibinfo{journal}{\emph{Journal of the London Mathematical Society}} \bibinfo{volume}{1}, \bibinfo{number}{1} (\bibinfo{year}{1967}), \bibinfo{pages}{226--228}.
\newblock


\bibitem[\protect\citeauthoryear{Bil{\`o}, Fanelli, Flammini, Monaco, and Moscardelli}{Bil{\`o} et~al\mbox{.}}{2019}]%
        {bilo2019optimality}
\bibfield{author}{\bibinfo{person}{Vittorio Bil{\`o}}, \bibinfo{person}{Angelo Fanelli}, \bibinfo{person}{Michele Flammini}, \bibinfo{person}{Gianpiero Monaco}, {and} \bibinfo{person}{Luca Moscardelli}.} \bibinfo{year}{2019}\natexlab{}.
\newblock \showarticletitle{Optimality and Nash Stability in Additively Separable Generalized Group Activity Selection Problems}. In \bibinfo{booktitle}{\emph{Proceedings of the 28th International Joint Conference on Artificial Intelligence}}. \bibinfo{pages}{102--108}.
\newblock


\bibitem[\protect\citeauthoryear{Bilò, Monaco, and Moscardelli}{Bilò et~al\mbox{.}}{2022}]%
        {bilo2022hedonic}
\bibfield{author}{\bibinfo{person}{Vittorio Bilò}, \bibinfo{person}{Gianpiero Monaco}, {and} \bibinfo{person}{Luca Moscardelli}.} \bibinfo{year}{2022}\natexlab{}.
\newblock \showarticletitle{Hedonic Games with Fixed-Size Coalitions}. In \bibinfo{booktitle}{\emph{Proceedings of the 36th AAAI Conference on Artificial Intelligence}}. \bibinfo{pages}{9287--9295}.
\newblock


\bibitem[\protect\citeauthoryear{Bogomolnaia and Jackson}{Bogomolnaia and Jackson}{2002}]%
        {bogomolnaia2002stability}
\bibfield{author}{\bibinfo{person}{Anna Bogomolnaia} {and} \bibinfo{person}{Matthew~O Jackson}.} \bibinfo{year}{2002}\natexlab{}.
\newblock \showarticletitle{The Stability of Hedonic Coalition Structures}.
\newblock \bibinfo{journal}{\emph{Games and Economic Behavior}} \bibinfo{volume}{38}, \bibinfo{number}{2} (\bibinfo{year}{2002}), \bibinfo{pages}{201--230}.
\newblock


\bibitem[\protect\citeauthoryear{Cechl{\'a}rov{\'a}, Romero-Medina, et~al\mbox{.}}{Cechl{\'a}rov{\'a} et~al\mbox{.}}{2001}]%
        {cechlarova2001stability}
\bibfield{author}{\bibinfo{person}{Katar{\i} Cechl{\'a}rov{\'a}}, \bibinfo{person}{Antonio Romero-Medina}, {et~al\mbox{.}}} \bibinfo{year}{2001}\natexlab{}.
\newblock \showarticletitle{Stability in Coalition Formation Games}.
\newblock \bibinfo{journal}{\emph{International Journal of Game Theory}} \bibinfo{volume}{29}, \bibinfo{number}{4} (\bibinfo{year}{2001}), \bibinfo{pages}{487--494}.
\newblock


\bibitem[\protect\citeauthoryear{Cseh, Fleiner, and Harj{\'a}n}{Cseh et~al\mbox{.}}{2019}]%
        {cseh2019pareto}
\bibfield{author}{\bibinfo{person}{{\'A}gnes Cseh}, \bibinfo{person}{Tam{\'a}s Fleiner}, {and} \bibinfo{person}{Petra Harj{\'a}n}.} \bibinfo{year}{2019}\natexlab{}.
\newblock \showarticletitle{Pareto Optimal Coalitions of Fixed Size}.
\newblock \bibinfo{journal}{\emph{arXiv preprint arXiv:1901.06737}} (\bibinfo{year}{2019}).
\newblock


\bibitem[\protect\citeauthoryear{Deineko and Woeginger}{Deineko and Woeginger}{2013}]%
        {deineko2013two}
\bibfield{author}{\bibinfo{person}{Vladimir~G Deineko} {and} \bibinfo{person}{Gerhard~J Woeginger}.} \bibinfo{year}{2013}\natexlab{}.
\newblock \showarticletitle{Two Hardness Results for Core Stability in Hedonic Coalition Formation Games}.
\newblock \bibinfo{journal}{\emph{Discrete Applied Mathematics}} \bibinfo{volume}{161}, \bibinfo{number}{13-14} (\bibinfo{year}{2013}), \bibinfo{pages}{1837--1842}.
\newblock


\bibitem[\protect\citeauthoryear{Deng and Papadimitriou}{Deng and Papadimitriou}{1994}]%
        {deng1994complexity}
\bibfield{author}{\bibinfo{person}{Xiaotie Deng} {and} \bibinfo{person}{Christos~H Papadimitriou}.} \bibinfo{year}{1994}\natexlab{}.
\newblock \showarticletitle{On the Complexity of Cooperative Solution Concepts}.
\newblock \bibinfo{journal}{\emph{Mathematics of operations research}} \bibinfo{volume}{19}, \bibinfo{number}{2} (\bibinfo{year}{1994}), \bibinfo{pages}{257--266}.
\newblock


\bibitem[\protect\citeauthoryear{Dreze and Greenberg}{Dreze and Greenberg}{1980}]%
        {dreze1980hedonic}
\bibfield{author}{\bibinfo{person}{Jacques~H Dreze} {and} \bibinfo{person}{Joseph Greenberg}.} \bibinfo{year}{1980}\natexlab{}.
\newblock \showarticletitle{Hedonic Coalitions: Optimality and Stability}.
\newblock \bibinfo{journal}{\emph{Econometrica: Journal of the Econometric Society}} (\bibinfo{year}{1980}), \bibinfo{pages}{987--1003}.
\newblock


\bibitem[\protect\citeauthoryear{Edmonds}{Edmonds}{1965}]%
        {edmons1965paths}
\bibfield{author}{\bibinfo{person}{Jack Edmonds}.} \bibinfo{year}{1965}\natexlab{}.
\newblock \showarticletitle{Paths, Trees, and Flowers}.
\newblock \bibinfo{journal}{\emph{Canadian Journal of Mathematics}}  \bibinfo{volume}{17} (\bibinfo{year}{1965}), \bibinfo{pages}{449--–467}.
\newblock


\bibitem[\protect\citeauthoryear{Flammini, Monaco, Moscardelli, Shalom, and Zaks}{Flammini et~al\mbox{.}}{2021}]%
        {flammini2021online}
\bibfield{author}{\bibinfo{person}{Michele Flammini}, \bibinfo{person}{Gianpiero Monaco}, \bibinfo{person}{Luca Moscardelli}, \bibinfo{person}{Mordechai Shalom}, {and} \bibinfo{person}{Shmuel Zaks}.} \bibinfo{year}{2021}\natexlab{}.
\newblock \showarticletitle{On the Online Coalition Structure Generation Problem}.
\newblock \bibinfo{journal}{\emph{Journal of Artificial Intelligence Research}}  \bibinfo{volume}{72} (\bibinfo{year}{2021}), \bibinfo{pages}{1215--1250}.
\newblock


\bibitem[\protect\citeauthoryear{Garey and Johnson}{Garey and Johnson}{1979}]%
        {garey1979computers}
\bibfield{author}{\bibinfo{person}{Michael~R Garey} {and} \bibinfo{person}{David~S Johnson}.} \bibinfo{year}{1979}\natexlab{}.
\newblock \bibinfo{booktitle}{\emph{Computers and Intractability: A Guide to the Theory of NP-completeness}}.
\newblock \bibinfo{publisher}{W. H. Freeman}.
\newblock


\bibitem[\protect\citeauthoryear{Hyafil and Rivest}{Hyafil and Rivest}{1973}]%
        {HR73}
\bibfield{author}{\bibinfo{person}{Laurent Hyafil} {and} \bibinfo{person}{Ronald~L Rivest}.} \bibinfo{year}{1973}\natexlab{}.
\newblock \bibinfo{booktitle}{\emph{Graph Partitioning and Constructing Optimal Decision Trees are Polynomial Complete Problems}}.
\newblock \bibinfo{type}{{T}echnical {R}eport} Rapport de Recherche no. 33. \bibinfo{institution}{IRIA -- Laboratoire de Recherche en Informatique et Automatique}.
\newblock


\bibitem[\protect\citeauthoryear{Li, Micha, Nikolov, and Shah}{Li et~al\mbox{.}}{2023}]%
        {li2023partitioning}
\bibfield{author}{\bibinfo{person}{Lily Li}, \bibinfo{person}{Evi Micha}, \bibinfo{person}{Aleksandar Nikolov}, {and} \bibinfo{person}{Nisarg Shah}.} \bibinfo{year}{2023}\natexlab{}.
\newblock \showarticletitle{Partitioning Friends Fairly}. In \bibinfo{booktitle}{\emph{Proceedings of the 37th AAAI Conference on Artificial Intelligence}}. \bibinfo{pages}{5747--5754}.
\newblock


\bibitem[\protect\citeauthoryear{Love}{Love}{1980}]%
        {10.2307/3615890}
\bibfield{author}{\bibinfo{person}{Eric~R. Love}.} \bibinfo{year}{1980}\natexlab{}.
\newblock \showarticletitle{64.4 Some Logarithm Inequalities}.
\newblock \bibinfo{journal}{\emph{The Mathematical Gazette}} \bibinfo{volume}{64}, \bibinfo{number}{427} (\bibinfo{year}{1980}), \bibinfo{pages}{55--57}.
\newblock


\bibitem[\protect\citeauthoryear{Olsen}{Olsen}{2009}]%
        {olsen2009nash}
\bibfield{author}{\bibinfo{person}{Martin Olsen}.} \bibinfo{year}{2009}\natexlab{}.
\newblock \showarticletitle{Nash Stability in Additively Separable Hedonic Games and Community Structures}.
\newblock \bibinfo{journal}{\emph{Theory of Computing Systems}} \bibinfo{volume}{45}, \bibinfo{number}{4} (\bibinfo{year}{2009}), \bibinfo{pages}{917--925}.
\newblock


\bibitem[\protect\citeauthoryear{Sless, Hazon, Kraus, and Wooldridge}{Sless et~al\mbox{.}}{2018}]%
        {sless2018forming}
\bibfield{author}{\bibinfo{person}{Liat Sless}, \bibinfo{person}{Noam Hazon}, \bibinfo{person}{Sarit Kraus}, {and} \bibinfo{person}{Michael Wooldridge}.} \bibinfo{year}{2018}\natexlab{}.
\newblock \showarticletitle{Forming k coalitions and facilitating relationships in social networks}.
\newblock \bibinfo{journal}{\emph{Artificial Intelligence}}  \bibinfo{volume}{259} (\bibinfo{year}{2018}), \bibinfo{pages}{217--245}.
\newblock


\bibitem[\protect\citeauthoryear{Sung and Dimitrov}{Sung and Dimitrov}{2010}]%
        {sung2010computational}
\bibfield{author}{\bibinfo{person}{Shao-Chin Sung} {and} \bibinfo{person}{Dinko Dimitrov}.} \bibinfo{year}{2010}\natexlab{}.
\newblock \showarticletitle{Computational Complexity in Additive Hedonic Games}.
\newblock \bibinfo{journal}{\emph{European Journal of Operational Research}} \bibinfo{volume}{203}, \bibinfo{number}{3} (\bibinfo{year}{2010}), \bibinfo{pages}{635--639}.
\newblock


\bibitem[\protect\citeauthoryear{Wright and Vorobeychik}{Wright and Vorobeychik}{2015}]%
        {wright2015mechanism}
\bibfield{author}{\bibinfo{person}{Mason Wright} {and} \bibinfo{person}{Yevgeniy Vorobeychik}.} \bibinfo{year}{2015}\natexlab{}.
\newblock \showarticletitle{Mechanism Design for Team Formation}. In \bibinfo{booktitle}{\emph{Proceedings of the 29th AAAI Conference on Artificial Intelligence}}. \bibinfo{pages}{1050--1056}.
\newblock


\end{thebibliography}

\appendix

\section{Additional Results}
\subsection{Tightness of the Approximation Ratio of the MnM Algorithm in the Unweighted Setting}
We show that our approximation ratio in the unweighted setting is tight. 

\begin{figure}[hbpt]
\centering
\begin{tikzpicture}[scale=0.65,node distance={12.5mm}, main/.style = {draw, circle, scale = 0.8}]
    \foreach \phi in {1,...,10}{
        \node[main] (v_\phi) at (72+ 360/10 * \phi:3cm){$v_{\phi}$};
    }
    \foreach \phi in {1,...,4}{
        \foreach \alpha in {6,...,10}{
            \draw (v_\phi) -- (v_\alpha);
        }
    }
    \foreach \alpha in {7,...,10}{
            \draw (v_5) -- (v_\alpha);
        }
    \foreach \i\j in {1/2,3/4,7/8,9/10}{
        \draw [blue, ultra thick, dotted](v_\i) -- (v_\j);
        }
    \foreach \i\j in {1/3,2/3,3/5,4/5}{
        \foreach \alpha in {\j,...,5}{
            \draw [blue, ultra thick] (v_\i) -- (v_\alpha);
        }
    }
    \foreach \i\j in {6/7,7/9,8/9}{
        \foreach \alpha in {\j,...,10}{
            \draw [blue, ultra thick] (v_\i) -- (v_\alpha);
        }
    }
    \draw [ultra thick, dotted](v_5) -- (v_6);
\end{tikzpicture}
\caption{An example of a graph in which $k=5$, and MnM achieves an approximation ratio of exactly $\frac{1}{4}$.}
\label{fig:complete_graph}
\end{figure}
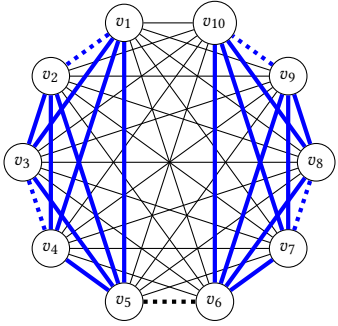

\begin{theorem}
The approximation ratio of MnM for the MaxUtil problem in the unweighted setting is tight. 
\end{theorem}
\begin{proof}
Given $k>2$, consider a complete graph of size $2k$. In this case, MnM finds a perfect matching in $M_1$, and thus the partition $P$ returned by MnM contains $k$ groups of $2$ nodes. That is, $u(P) = 2k$. On the other hand, an optimal $k$-bounded partition $Opt$ consists of $2$ Cliques of size $k$, and thus $u(Opt) = 2k(k-1)$. That is, MnM provides an approximation of exactly $\frac{1}{k-1}$.
\end{proof}
Figure \ref{fig:complete_graph} presents a case where $k=5$, and $G$ is a complete graph with $10$ nodes. Here, $P=\{\{v_1,v_2\},\{v_3,v_4\},\{v_5,v_6\},\{v_7,v_8\},\{v_9,v_{10}\}\}$, as shown in the dotted lines, and thus $u(P) = 10$ . However, $Opt = \{\{v_1,v_2,v_3,v_4,v_5\}, \{v_6,v_7,v_8,v_9,v_{10}\}\}$, as shown in the blue lines, and $u(Opt) = 40$.

\subsection{MaxUtil with Distributed Partition in Unweighted Setting}

We analyze a procedure that attempts to model the behavior of the agents when there is no central mechanism that determines the partition. Assume that the agents are split up arbitrarily but maximally, i.e., in a way that no two coalitions can joint together such that the social welfare will be higher. We call this procedure $Arbmax$. 
Without loss of generality, we assume that every set $S \in Arbmax$ is a connected component.
We show that $\frac{1}{k}$ is an upper bound on the approximation ratio that $Arbmax$ may guarantee.
\begin{theorem}
For any $k$, $Arbmax$ cannot guarantee an approximation ratio better than $\frac{1}{k}$
\end{theorem}
\begin{proof}
Given $k$, consider the following graph $G$. There are $k$ distinguished nodes, $v_1,\ldots,v_k$, with the edges  $(v_i,v_{i+1})\in E$ for $i=1,\ldots,k-1$. Each distinguished node $v_i$ has $k-1$ additional neighbors that are connected only to $v_i$, i.e., $v_i$ is the internal node of a star graph with $k-1$ leaves. Clearly, $Opt$ consists of $k$ coalitions, where each coalition consists of a star graph. Thus, $u(Opt)=2k(k-1)$. On the other hand, $Arbmax$ may partition the graph such that the distinguished nodes $v_1,\ldots,v_k$ are in the same coalition. Since there are no edges between two undistinguished nodes, the social welfare of this partition is $2(k-1)$. Therefore, $Arbmax$ cannot guarantee an approximation ratio better than $\frac{1}{k}$.
\end{proof}
Figure \ref{fig:Arbmax} presents a case where $k=5$, and $Arbmax$ may provide only a $\frac{1}{5}$ of the social welfare provided by an optimal solution. Here, $Arbmax$ may return the partition 
$P' = \{\{v_1, v_2, v_3, v_4, v_5\}, \{v_6\}, \{v_7\}, $ $ \{v_8\}, \{v_9\}, \{v_{10}\}, \{v_{11}\}, \{v_{12}\}, \{v_{13}\}, \{v_{14}\}, \{v_{15}\}, \{v_{16}\},$ $ \{v_{17}\}, \{v_{18}\},$ \, $ \{v_{19}\}, \{v_{20}\}, \{v_{21}\}, \{v_{22}\}, \{v_{23}\}, \{v_{24}\}, \{v_{25}\}\}$ and thus $u(P') = 8$, while $Opt = \{\{v_1, v_6, v_7, v_8, v_9\}, \{v_2, v_{10}, v_{11}, v_{12}, v_{13}\}, \{v_3, v_{14}, v_{15}, $ \, $ v_{16}, v_{17}\}, \{v_4, v_{18}, v_{19}, v_{20}, v_{21}\}, \{v_5, v_{22}, v_{23}, v_{24}, v_{25}\}\}$ and therefore $u(Opt) = 40$.

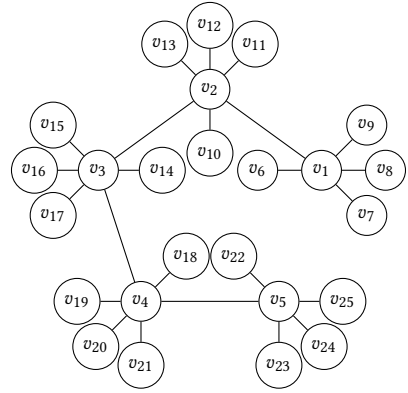
\begin{figure}
\centering
\begin{tikzpicture}[scale=0.65,node distance={10mm}, main/.style = {draw, circle, scale = 0.85}]
    \foreach \phi in {1,...,5}{
        \node[main] (\phi) at (-54 + 360/5 * \phi:2.4cm){$v_{\phi}$};
    }
    \node[main] (6) [left of=1] {$v_6$}; 
    \node[main] (7) [below right of=1] {$v_7$}; 
    \node[main] (8) [right of=1] {$v_8$}; 
    \node[main] (9) [above right of=1] {$v_9$}; 
    \node[main] (10) [below of=2] {$v_{10}$}; 
    \node[main] (11) [above right of=2] {$v_{11}$}; 
    \node[main] (12) [above of=2] {$v_{12}$}; 
    \node[main] (13) [above  left of=2] {$v_{13}$};
    \node[main] (14) [right of=3] {$v_{14}$}; 
    \node[main] (15) [above left of=3] {$v_{15}$}; 
    \node[main] (16) [left of=3] {$v_{16}$}; 
    \node[main] (17) [below left of=3] {$v_{17}$}; 
    \node[main] (18) [above right of=4] {$v_{18}$}; 
    \node[main] (19) [left of=4] {$v_{19}$}; 
    \node[main] (20) [below left of=4] {$v_{20}$}; 
    \node[main] (21) [below of=4] {$v_{21}$}; 
    \node[main] (22) [above left of=5] {$v_{22}$}; 
    \node[main] (23) [below of=5] {$v_{23}$}; 
    \node[main] (24) [below right of=5] {$v_{24}$}; 
    \node[main] (25) [right of=5] {$v_{25}$}; 
    \draw (1) -- (2);
    \draw (2) -- (3);
    \draw (3) -- (4);
    \draw (4) -- (5);
    \draw (1) -- (6);
    \draw (1) -- (7);
    \draw (1) -- (8);
    \draw (1) -- (9);
    \draw (2) -- (10);
    \draw (2) -- (11);
    \draw (2) -- (12);
    \draw (2) -- (13);
    \draw (3) -- (14);
    \draw (3) -- (15);
    \draw (3) -- (16);
    \draw (3) -- (17);
    \draw (4) -- (18);
    \draw (4) -- (19);
    \draw (4) -- (20);
    \draw (4) -- (21);
    \draw (5) -- (22);
    \draw (5) -- (23);
    \draw (5) -- (24);
    \draw (5) -- (25);
\end{tikzpicture}
\caption{A case where $k=5$, and $Arbmax$ may provide only a $\frac{1}{5}$ of the social welfare provided by an optimal solution.}
\label{fig:Arbmax}
\end{figure}

\subsection{Worst Case Example for the MnM Algorithm in the Weighted Setting}
We show that in the weighted setting, MnM cannot guarantee a better approximation ratio. Indeed, Figure \ref{fig:MnM} shows an example where $k=3$ and MnM $=\{\{v_1,v_2\}, \{v_3, v_4\}, \{v_5, v_6\}\}$. In this case, the outcome of MnM provides only a $\frac{2+\epsilon}{6}$ of the social welfare provided by the optimal solution $\{\{v_1, v_2, v_3\}, \{v_4, v_5, v_6\}\}$. Therefore, for any $\epsilon>0$, when $k=3$, MnM cannot provide an approximation that is better than $\frac{1}{k}+\epsilon$.

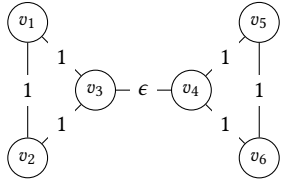
\begin{figure}
\centering
\begin{tikzpicture}[scale=0.65,node distance={15mm}, main/.style = {draw, circle, scale = 0.85}]
    \node[main] (1) {$v_1$};
    \node[main] (3) [below right of=1] {$v_3$}; 
    \node[main] (2) [below left of=3] {$v_2$}; 

    \node[main] (4) [right of=3] {$v_4$};     
    \node[main] (5) [above right of=4] {$v_5$}; 
    \node[main] (6) [below right of=4] {$v_6$}; 
    
    \tikzset{mystyle/.style={color=black}} 
    \tikzset{every node/.style={fill=white}}
    \path   (1)     edge [mystyle]    node  {$1$} (2)
            (2)     edge [mystyle]    node  {$1$} (3)
            (1)     edge [mystyle]    node  {$1$} (3)
            (4)     edge [mystyle]    node  {$1$} (5)
            (4)     edge [mystyle]    node  {$1$} (6)
            (5)     edge [mystyle]    node  {$1$} (6)
            (3)     edge [mystyle]    node  {$\epsilon$} (4);
\end{tikzpicture}
\caption{A case where $k=3$, and MnM provides only a $\frac{2+\epsilon}{6}$ of the social welfare provided by an optimal solution.}
\label{fig:MnM}
\end{figure}

\subsection{Approximation for the Core in the Unweighted Setting}
We now investigate additive and multiplicative approximations of the core, which are defined as follows. 
\begin{definition}[Additive approximation]
A $k$-bounded coalition $S$ is said to \emph{$\epsilon_a$-strongly block} a $k$-bounded partition $P$ if it improves the utility of each of its members by more than an additive factor of $\epsilon_a$. That is, for every $v \in S$, $W(v,S) > u(v,P) + \epsilon_a$.
A $k$-bounded partition $P$ is in the \emph{$\epsilon_a$-core} if it does not have any $\epsilon_a$-strongly blocking $k$-bounded coalitions.
\end{definition}

The \emph{$\epsilon_m$-core}, which is the multiplicative approximation of the core is defined similarly. That is, a $k$-bounded coalition $S$ is said to \emph{$\epsilon_m$-strongly block} a $k$-bounded partition $P$ if for every $v \in S$, $W(v,S) > \epsilon_m \cdot u(v,P)$.



We now show that for $k>3$ the $\epsilon_a$-core, for $\epsilon_a = \lfloor \frac{k}{2} \rfloor -1$, is never empty. We present Algorithm \ref{alg:Addcore}, a polynomial time algorithm that finds a $k$-bounded partition $P$ in the $\epsilon_a$-core. The algorithm begins with all agents in singletons and iteratively considers for each $k$-bounded coalition whether it $\epsilon_a$-strongly blocks the current partition.

\begin{algorithm}[ht]
    \caption{Finding a $k$-bounded partition in the $\epsilon_a$-core}
    \label{alg:Addcore}
    \SetAlgoLined
    \textbf{Input}:
    A graph $G(V,E)$\\
    A limit $k$\\
    An additive factor $\epsilon_a$\\
    \KwResult{A $k$-bounded partition $P$ of $V$ in the $\epsilon_a$-core.} 
    $P \leftarrow \{\{v\}$ for every $v \in V\}$\\
    $V' \leftarrow V$\\
    outerLoop:\\ 
    \label{Addcore:outerLoop}
        \For{$S \subset V'$, such that $1<|S|\leq k$}{
            \If{$\forall v \in S, W(v ,S) > u(v,P) + \epsilon_a$}{
            \label{Addcore:if}
                $P \leftarrow P^{-S}$ \\
                \If{$\forall v \in S, W(v ,S) \geq k-1-\epsilon_a$}{
                    $V' \leftarrow V' \setminus S$\\ \label{Addcore:line:remove}
                }
                \textbf{goto} outerLoop
            }
        }
        \textbf{return} P
\end{algorithm}

\begin{theorem}
For $\epsilon_a = \lfloor \frac{k}{2} \rfloor -1$, there always exists a $k$-bounded partition in the $\epsilon_a$-core, and it can be found in polynomial time.
\end{theorem}

\begin{proof}
Consider Algorithm \ref{alg:Addcore}.
Note that for every $k$-bounded partition $P$, and $S \in P$, if for every $v \in S$, $W(v ,S) \geq k-1-\epsilon_a$ then every $v \in S$ cannot belong to any $\epsilon_a$-strongly blocking $k$-bounded coalition.
Therefore, Algorithm~\ref{alg:Addcore} removes such vertices from $V'$ (in line~\ref{Addcore:line:remove}).
Clearly, if Algorithm \ref{alg:Addcore} terminates, the $k$-bounded partition $P$ is in the $\epsilon_a$-core.
We now show that for $\epsilon_a = \lfloor \frac{k}{2} \rfloor -1$, Algorithm \ref{alg:Addcore} must always terminate, and it runs in polynomial time.

First, we show that $u(P)$ never decreases.
The algorithm initiates a new iteration (line  \ref{Addcore:outerLoop}) whenever the if statement in line \ref{Addcore:if} is true. This can happen only if for every $v$ in the blocking coalition $S$, $u(v,P)$ is less than $k-1-\lfloor \frac{k}{2} \rfloor + 1$ (since $W(v ,S)$ is at most $k-1$). Therefore, $u(v,P) < k-\frac{k}{2}+\frac{1}{2} = \frac{k}{2}+\frac{1}{2}$.
Since $u(v,P)$ is a natural number, $u(v,P) \leq \frac{k}{2}-\frac{1}{2}$.
When $S$ breaks off, the social welfare decreases by at most $2 \cdot \sum\limits_{v\in S} u(v,P)$, and increases by at least $\sum\limits_{v\in S} (u(v,P)+\lfloor \frac{k}{2} \rfloor)$.
Since $\sum\limits_{v\in S} (u(v,P)+\lfloor \frac{k}{2} \rfloor) \geq \sum\limits_{v\in S} (u(v,P)+ \frac{k}{2} - \frac{1}{2}) \geq \sum\limits_{v\in S} (u(v,P)+ u(v,P)) = 2 \cdot \sum\limits_{v\in S} u(v,P)$, then $u(P)$ never decreases.

Next, we show that if $u(P)$ remains the same, then vertices are removed from further consideration (in line~\ref{Addcore:line:remove}). Observe that $u(P)$ remaining the same entails that $2 \cdot \sum\limits_{v\in S} u(v,P) = \sum\limits_{v\in S} (u(v,P)+\lfloor \frac{k}{2} \rfloor)$.
That is, $\sum\limits_{v\in S} u(v,P) = \sum\limits_{v\in S} \lfloor \frac{k}{2} \rfloor$. Recall that for every $v \in S$, $u(v,P) \leq \frac{k}{2}-\frac{1}{2}$ and note that $\frac{k}{2}-\frac{1}{2} \leq \lfloor \frac{k}{2} \rfloor$. Therefore, if $u(P)$ remains the same, then for every $v \in S$, $u(v,P) = \frac{k}{2}-\frac{1}{2}$. 
Now, since for every $v \in S$,  $W(v,S) > u(v,P) + \lfloor \frac{k}{2} \rfloor -1$ and $\epsilon_a \geq 1$, then $W(v,S) > \frac{k}{2} - \frac{1}{2} + \frac{k}{2} - \frac{1}{2} - 1 = k-2 \geq k-1-\epsilon_a$.
Therefore, all $v \in S$ are removed from further consideration. 

Overall, either $u(P)$ has increased by at least $2$ or vertices are removed from further consideration.
Since $u(P)$ is bounded by $2|E|$ and the number of vertices is finite, the algorithm must terminate after at most $|E|+\frac{|V|}{k}$ iterations.
\end{proof}

Next, we show that for $k>3$ the $\epsilon_m$-core, for $\epsilon_m = 2$, is never empty. We use Algorithm \ref{alg:Addcore}, with the following changes:
%
%
\begin{itemize}
    \item The input of the algorithm is $\epsilon_m$ instead of $\epsilon_a$ (in line $3$).
    \item In line $8$, we check if for every $v$ in $S$, $W(v ,S) > \frac{u(v, P)}{\epsilon_m}$.
    \item In line $10$, we check if for every $v$ in $S$, $W(v,S) \geq \frac{k-1}{\epsilon_m}$.
\end{itemize}
We show that this algorithm finds a $k$-bounded partition $P$ in the $\epsilon_m$-core in polynomial time.

\begin{theorem}
\label{theorem:9}
For $\epsilon_m = 2$, there always exists a $k$-bounded partition in the $\epsilon_m$-core, and it can be found in polynomial time.
\end{theorem}
\begin{proof}
Clearly, if the algorithm terminates, the $k$-bounded partition $P$ is in the $\epsilon_m$-core.
We show that $u(P)$ always increases.
The algorithm initiates a new iteration whenever $\forall v \in S, W(v ,S) > \frac{u(v, P)}{\epsilon_m}$, which can happen only if $S$ breaks off.
When $S$ breaks off, the social welfare decreases by at most $2 \cdot \sum\limits_{v\in S} u(v,P)$, and increases by at least $\sum\limits_{v\in S} (2 \cdot u(v,P)+1)$.
Since $\sum\limits_{v\in S} (2 \cdot u(v,P)+1) > 2 \cdot \sum\limits_{v\in S} u(v,P)$, then $u(P)$ always increases.
\end{proof}

\subsection{Nash Stablity}

Due to \cite{bogomolnaia2002stability}, for any $(G,k)$ a Nash stable $k$-bounded partition exists.
In this section, we present Algorithm \ref{alg:NS}, a polynomial time algorithm that finds such a partition in the unweighted setting. The algorithm begins with all agents in singletons and iteratively considers for each agent whether it may benefit from leaving her coalition and joining a coalition of size of at most $k-1$.

\begin{algorithm}[ht]
    \caption{Finding a Nash stable $k$-bounded partition}
    \label{alg:NS}
    \SetAlgoLined
    \textbf{Input}:
    A graph $G(V,E)$ and a limit $k$\\
    \KwResult{A $k$-bounded Nash Stable partition $P$ of $V$.} 
    $P \leftarrow \{\{v\}$ for every $v \in V\}$\\
    
    outerLoop:\\ \label{NS:outerLoop} 
        \For{$v \in V$}{
            \For{$S \in P$}{
                \If{$W(v,S \cup \{v\}) > u(v,P)$  AND $|S| \leq k-1$}{
                \label{NS:if}
                    $P \leftarrow P^{-S \cup \{v\}}$ \\
                    \textbf{goto} outerLoop
                }
            }
        }
        \textbf{return} P
\end{algorithm}

\begin{theorem}
There always exists a $k$-bounded Nash stable partition, and it can be found in polynomial time.
\end{theorem}

\begin{proof}
Consider Algorithm \ref{alg:NS}.
Clearly, when Algorithm \ref{alg:NS} terminates, the partition $P$ is a $k$-bounded Nash Stable partition.
We now show that Algorithm \ref{alg:NS} must always terminate, and it runs in polynomial time.
Returning to line \ref{NS:outerLoop} occurs only when the if statement in line \ref{NS:if} is true, which entails that $u(P)$ has increased. Since any increase in $u(P)$ must be by multiples of $2$, and since $u(P)$ is bounded by $2|E|$, the algorithm must terminate after at most $|E|$ iterations.
\end{proof}

\end{document}